\documentclass[final,3p,times,authoryear]{elsarticle}

\RequirePackage[OT1]{fontenc}
\usepackage{amsthm,amsmath,natbib}
\RequirePackage[colorlinks,citecolor=blue,urlcolor=blue]{hyperref}
\usepackage{amssymb}
\usepackage{graphicx,overpic,amsmath}
\usepackage{bm}
\usepackage{upgreek}
\usepackage{url}
\usepackage{color}
\usepackage[overload]{textcase}
\usepackage[ruled]{algorithm2e}
\usepackage{makecell}
\usepackage{natbib}
\usepackage{multirow}

\usepackage{amsthm}
\usepackage{amsmath}
\usepackage{amssymb}  
\usepackage{dsfont}   
\usepackage{multirow} 
\usepackage{tabularx} 
\usepackage{booktabs} 
\usepackage{pdflscape} 

\numberwithin{equation}{section}

\usepackage{graphicx,psfrag,epsf}
\usepackage{enumerate}
\usepackage{natbib}
\usepackage{url} 

\newtheorem{theorem}{Theorem}[section]
\newtheorem{lemma}{Lemma}[section]

\newtheorem{example}{Example}

\newcommand{\ba}{\begin{array}}
\newcommand{\ea}{\end{array}}
\newcommand{\bt}{\begin{tabular}}
\newcommand{\et}{\end{tabular}}
\newcommand{\btb}{\begin{table}}
\newcommand{\etb}{\end{table}}
\newcommand{\bc}{\begin{center}}
\newcommand{\ec}{\end{center}}
\newcommand{\bea}{\begin{eqnarray}}
\newcommand{\eea}{\end{eqnarray}}
\newcommand{\Bea}{\begin{eqnarray*}}
\newcommand{\Eea}{\end{eqnarray*}}
\newcommand{\beq}{\begin{equation}}
\newcommand{\eeq}{\end{equation}}

\newcommand{\p}{{\rm I}\kern-0.18em{\rm P}}
\newcommand{\1}{{\rm 1}\kern-0.24em{\rm I}}
\newcommand{\E}{{\rm I}\kern-0.18em{\rm E}}
\newcommand{\R}{{\rm I}\kern-0.18em{\rm R}}

\def \bfm#1{\mbox{\boldmath$#1$}}
\def \es2{$E(s^2)$}

 \def \x {{\bfm x}}

\def \R {{\bfm R}}

\def \0{{\bf 0}}

 \def \s2{{\sigma^2}}

\begin{document}

\begin{frontmatter}



\title{Model-free Feature Screening via Revised Chatterjee’s Rank Correlation for Ultra-high Dimensional Censored Data} 


\author[label1]{Shuya Chen} 

\author[label2]{Heng  Peng} 

\affiliation[label1]{organization={Guangdong Provincial / Zhuhai Key Laboratory of Interdisciplinary Research and Application for Data Science, Faculty of Science and Technology, Beijing Normal-Hong Kong Baptist University},
            city={Zhuhai},
            state={Guangdong Province},
            country={China}}

\author{Min Zhou\corref{cor1}\fnref{label1}}
\ead{minzhou@bnbu.edu.cn}
\cortext[cor1]{Corresponding author}

\affiliation[label2]{organization={Department of  Mathematics, Hong Kong Baptist University},
            state={Hong Kong},
            country={China}}

\begin{abstract}
In large-scale biomedical research, it's common to gather ultra-high dimensional data that includes right-censored survival times. 
Feature screening has emerged as a crucial statistical technique for handling such data. In this paper, we introduce a straightforward and robust feature screening approach, leveraging the modified Chatterjee's rank correlation, suitable for a broad range of survival models. With reasonably mild regularity assumptions, we establish the properties of sure screening and ranking consistency. The computation involved in our proposed method is quite direct and simple. Through simulation studies and real gene expression data analysis, we demonstrate the superior efficacy of our proposed approach.
\end{abstract}




\end{frontmatter}




\section{Introduction}\label{sec:intro}
As data acquisition techniques continue to advance rapidly, ultra-high dimensional data are increasingly common across various scientific disciplines, including genetic engineering, finance, public health, molecular biology, clinical genomics, and brain imaging research, among others. It's not uncommon to encounter millions of potential predictors related to the response in a typical genetic microarray study, while the sample size remains in the hundreds. To create a more parsimonious model, penalization approaches such as the least absolute shrinkage and selection operator (LASSO) \citep{tibshirani1996regression}, smoothly clipped absolute deviation (SCAD) \citep{fan2001variable},  the adaptive LASSO \citep{zou2006adaptive}, the Group LASSO \citep{yuan2006model}, the Dantzig selector \citep{candes2007dantzig}, and the minimax concave penalty (MCP) \citep{zhang2010nearly} have been extensively explored over the past two decades. However, the computational complexity inherent in these methods makes them difficult to apply directly to data where the dimensionality $p$ increases at a exponential rate relative to the sample size $n$ (e.g., $p=\exp \left(n^\alpha\right)$ with $\alpha>0$). These methods face the simultaneous challenges of computational efficiency, statistical accuracy, and algorithmic stability \citep{fan2009ultrahigh}.

One  simple and effective  method that addresses the aforementioned trio of challenges, is the marginal independent learning \citep{fan2008sure}. The sure independence screening (SIS) method  was  proposed by \cite{fan2008sure}  in linear regression by ranking the importance of each candidate variable by the marginal Pearson correlation with the response. They further demonstrated that this method possesses a sure screening property--meaning that all significant features are identified with a probability approaching 1. Motivated by \cite{fan2008sure}, a substantial amount of research has been focused on the study of feature screening in recent years, primarily due to its computational efficiency. See   \cite{fan2011nonparametric}, \cite{fan2014nonparametric}, \cite{fan2009ultrahigh}, \cite{fan2010sure},   \cite{hall2009using}, \cite{li2012robust}, \cite{li2012feature},  \cite{liu2014feature}, \cite{wang2009forward},  \cite{xu2014sparse},  \cite{zhong2016regularized},  \cite{zhou2023bolt}, \cite{zhu2011model}, \cite{yang2016feature}, \cite{yang2018feature} among many others.

In many large-scale biomedical studies, the presence of a response that is time to event and right censored, introduces additional complexities and significant challenges when performing feature screening on ultra-high dimensional data.  Some works were published to discuss this issue.  To name a few, 
\cite{fan2010high} extended the Sure Independence Screening (SIS) to the Cox Proportional Hazards model with an iterative version available. \cite{zhao2012principled} developed the principled sure independence screening (PSIS) based on the standardized marginal maximum partial likelihood estimators for the Cox model. \cite{gorst2013independent} proposed a survival variant of independent screening based on the feature aberration at survival times (FAST) statistic for single-index hazard rate models. \cite{song2014censored} introduced the censored rank independence screening (CRIS) of high-dimensional survival data based on an inverse probability censoring weighted version of Kendall' $\tau$. The study by \cite{zhang2017correlation} examined a procedure known as Correlation Rank Sure Independent Screening (CRSIS), which measures the correlation between each predictor and the cumulative distribution function of the response.
\cite{zhou2017model} extended the non-parametric SIRS to the censored SIRS (CSIRS) using the inverse probability censoring weighting concept.  \cite{chen2018robust} introduced the robust censored distance correlation screening (RCDCS) and composite robust censored distance correlation screening (CRCDCS). \cite{hao2019robust} extended the robust censored sure independent screening (rCSIS). \citep{lin2018model} into the robust censored feature ranking and screening (rCFRS) by the Pearson’s correlation of the distribution of each predictor and the survival function of the survival time. \cite{zhang2021model} proposed the censored distance correlation screening (CDCS) based on the distance correlation after standardizing each covariate and a new variable by combining the information of the observed time and the censoring indicator. \cite{song2022model} generalized the method of \cite{zhu2011model} to the right censored response situation and proposed the weighted sure independent ranking and screening (WSIRS). \cite{li2023model} developed HSIC-SIS by ultrahigh-dimensional survival data using the Hilbert-Schmidt Independence Criterion (HSIC).

In this paper, we proposed a new model-free and robust screening method named as XIMSIS for ultra-high dimensional right-censored survival data. Our method is based on a new correlation--Chatterjee’s correlation $\xi$ \citep{chatterjee2021new, lin2023boosting}. They proved that $\xi$ is indeed a dependence criterion, i.e., $\xi$ is zero if and only if the random variables are independent. 
It presents an attractive option for assessing the strength of bivariate associations, particularly when identifying perfect functional dependence, as discussed in \cite{cao2020correlations}. This advantageous characteristic encourages us to employ the estimator $\xi_{n, M}$ of $\xi$ to filter irrelevant predictors within ultra-high dimensional survival data.
 Compared with previous screening procedures, this novel method boasts several notable benefits. 
 First, Chatterjee’s correlation $\xi$ is a good measure of perfect nonlinear functional dependence. It is zero means that the random variables are independent. 
 Second, the  new screening procedure XIM-SIS based on $\xi_{n,M}$ is robust to the model misspecification and also robust in the presence of outliers or extreme values because the estimator $\xi_{n,M}$  is  one type of rank correlation.
 Third, the corresponding calculation of the proposed procedure  is very simple and fast.
 Finally, this new method maintains the sure screening property under certain regular conditions. Both Monte Carlo simulations and real data analyses demonstrate the effective performance of the proposed procedure. These benefits significantly enhance its practical applicability.

The rest of this paper is organized as follows. In Section 2, we  introduce the new feature screening method for ultra-high dimensional censored data based on the revised Chatterjee’s rank correlation. Section 3 discusses the theoretical properties of the proposed methods. In Section 4 and 5, we provide numerical studies and real data examples to demonstrate the effectiveness of our approach.  We conclude this work with a short discussion in  Section 6. The Appendix contains the regularity conditions and detailed proofs of the main theorems.

\section{Methodologies and main results}
In this section, we propose a robust censored feature screening method through the Chatterjee’s correlation.
Suppose that $\x =\left(X_1, \ldots, X_{p}\right)^{\top} $ be  the $p$-dimensional vector of features, $C$ denote the censoring time, and $T$ denotes the event time. Define $Y=\min (T, C)$, $\Delta=\1(T \leq C)$, where $\1(\cdot)$ denotes the indicator function.  
The observed data are independent and identically distributed copies of $(\x, Y,\Delta)$ and are denoted by $(\x_i , Y_i , \Delta_i)$ for $i=1,\dots,n$, where $\x_i =(X_{1i},\ldots,X_{pi})$.
Throughout the paper, it is assumed that the censoring time $C$ is independent of the event time $T$ and the covariates  $\x$. 
The conditional distribution function of $T$ given $\x$ is denoted by $F(t \mid \x)$. Without model assumption, we define the active set as
$$\mathcal{A}=\left\{k: F(t \mid \x) \text { depends on } X_k \text { for some } t \in \Omega_T\right\}$$
and the inactive sets is defined as $\mathcal{A}^c=\left\{1,2, \ldots, p\right\} \backslash \mathcal{A}$, where $\Omega_T$ is the support set of response $T$. 
The variable $X_k$ is called an active predictor if $k \in \mathcal{A}$ while an inactive predictor if $k \in \mathcal{A}^c$. Our goal is to identify the important  covariates through marginal screening.

\subsection{The correlation $\xi$ and revised Chatterjee’s   rank correlation $\xi_{n,M}$}
To facilitate the presentation, we review the Chatterjee’s  correlation $\xi$ and  and its estimation firstly.
Consider a pair of random variables $(U, V)$, where $V$ is not constant. Let $F_{U, V}$ be their joint cumulative distribution function, and $F_U$ and $F_V$ represent their respective marginal cumulative distribution functions respectively.  
\cite{dette2013copula} introduced their Dette-Siburg-Stoimenov coefficient $\xi$  between $U$ and $V$, formulated as:
\begin{align}\label{xi}
	\xi(U, V) \equiv\frac{\int \operatorname{Var}[E\{\1(V \geqslant v) \mid U\}] \mathrm{d} F_V(v)}{\int \operatorname{Var}\{\1(V \geqslant v)\} \mathrm{d} F_V(v)}
\end{align} 
As a measure of correlation, the correlation $\xi$ has an attractive property, i.e., $\xi(U, V)=0$ if and only if $U$ and $V$ are independent. This makes the  correlation $\xi$ particularly suitable for variable screening of ultra-high dimensional data.
The corresponding  estimator called Chatterjee’s rank correlation $\xi_n$  of  $\xi$ was also proposed by \cite{chatterjee2021new}. 

 However, \cite{lin2023boosting} pointed out  that independence tests based on Chatterjee’s rank correlation are unfortunately rate inefficient against various local alternatives. Therefore,  they  proposed a new efficient and consistent estimator--the revised Chatterjee’s rank correlation $\xi_{n,M}$.  Assume that  $\left(U_1, V_1\right), \ldots,\left(U_n, V_n\right)$ are $n$ independent copies of $(U, V)$.  Define the rank $R_i$ of $V_i$ among the set $\left\{V_1, \ldots, V_n\right\}$ as
$R_i \equiv \sum_{j=1}^n \1\left(V_j \leqslant V_i\right)$.
Define $j_m(i)$ as the index $j \in\{1, \ldots, n\}$  such that for any $i \in\{1, \ldots, n\}$ and $m \in\{1, \ldots, n\}$, 
$\sum_{k=1}^n \1\left(U_i<U_k \leqslant U_j\right)=m$
if such $j$ exists; otherwise, $j_m(i)=i$.
 Then, \cite{lin2023boosting}   defined the revised Chatterjee's rank correlation $\xi_{n, M}$ as
\begin{align}\label{revisedxi}
	\xi_{n, M}=\xi_{n, M}\left(\left[\left(U_i, V_i\right)\right]_{i=1}^n\right) \equiv-2+\frac{6 \sum_{i=1}^n \sum_{m=1}^M \min \left\{R_i, R_{j_m(i)}\right\}}{(n+1)\{n M+M(M+1) / 4\}}
\end{align}
where $M\in\{1, \ldots, n\}$ represents the number of right nearest neighbors. Obviously, different values of $M$ may lead to varying levels of final accuracy, and  the selection of  $M$   has been discussed by  \cite{lin2023boosting}.   And it was demonstrated that $M$ is of the order of $\sqrt{n}$. In the numerical studies,  we select  the  value  $M$  to be $c\sqrt{n}$, where $c$ is some positive constant.

\subsection{Revised Chatterjee’s rank correlation screening}

In fact, it is unfeasible to measure the correlation between the survival time and each covariate through the correlation  $\xi$ (\ref{xi}) for  right censored data. Furthermore, many covariates may follow heavy-tailed distributions and contain outliers. To assess the correlation in this circumstance, we suggest to use the  Chatterjee’s  correlation of $X_k$ and $T$ by  $\xi\left(F_{k}(X_k), S(T)\right)$, where $F_{k}(x)=P(X_k\leq x)$ and $S(t)=P(T\geq t)$. Unlike most coefficients,  the correlation $\xi$  is not symmetric, that is, $\xi(U,V)\neq \xi(V,U)$. $\xi(U, V)$ is an estimator of whether $V$ is a measurable function of $U$,
 and if we want to measure whether $U$ is a function of $V$, $\xi_n(V, U)$ should be used instead of $\xi_n(U, V)$. A usage of the symmetrized Chatterjee’s correlation coefficient $\max \{\xi(U, V), \xi(V, U)\}$ is also recommended by \cite{chatterjee2021new}, and some of its properties have been researched by \cite{zhang2023asymptotic}.

Thus, we take maximum of  $\xi\left(F_{k}(X_k), S(T)\right)$ and $\xi\left( S(T), F_{k}(X_k)\right)$ as a measure of correlation between $X_k$ and $T$, i.e., 
$$
	\omega_k = \max\left\{ \ \xi\left(F_{k}(X_k), S(T)\right), \ \xi\left(S(T), F_{k}(X_k)\right) \right\}, \quad k=1, \ldots, p.
$$
Hence,  $\omega_k$ is symmetric, and the independence of $X_k$ and $T$ implies $\omega_k=0$. 

We then apply the revised Chatterjee’s rank correlation (\ref{revisedxi}) to estimate $\omega_k$, that is, 
\begin{align}
	\hat{\omega}_k &= \max\left\{  \xi_{n,M}\left(\hat{F}_{k}(X_k), \hat{S}(T)\right), \ \xi_{n,M}\left(\hat{S}(T), \hat{F}_{k}(X_k)\right) \right\}
\end{align}
where $\hat{F}_{k}(x)=\frac1  n \sum_{i=1}^n \1\left(X_{i k} \leq x\right)$, $\hat{S}(\cdot)$ is the Kaplan-Meier estimate (KM) of $T$ with form
\begin{align}
	\hat{S}(t) \equiv \prod_{i: Y_i \leq t}\left[1-\frac{\Delta_i}{\sum_{j=1}^n \1\left(Y_j \geq Y_i\right)}\right] . \notag
\end{align}
We consider  to rank the  predictors by  $\omega_k$ as a measure of marginal utility. A high value of $\hat{\omega}_k$ indicates a stronger correlation between the response $T$ and the  covariate $X_k$.  With a pre-specified $\gamma$, we select a set of variables
$$
	\hat{\mathcal{A}}=\left\{k: \hat{\omega}_k \geq \gamma, \text { for } 1 \leq k \leq p\right\}.
$$
We call this screening procedure as XIM-SIS.
Next, we study the theoretical properties of the proposed screening procedure XIM-SIS. The following assumptions are needed to establish the ranking consistency property and sure screening property.

\begin{itemize}
	\item[(C1)] There exist positive constants $v$ and $\eta$ such that $P(t \leq T \leq C) \geq \eta$, where $t \in\left(0, v\right]$ and the $v$ is the maximum follow-up time.
	\item[(C2)]  The minimum of marginal signals of truly important covariates satisfies
	\begin{align}
		\min _{k \in \mathcal{A}} \omega_k \geqslant 2 c n^{-\kappa}
	\end{align}
	for some constants $c>0$ and $0 \leqslant \kappa<1 / 2$. 
\end{itemize}

The assumption (C1) is common in the survival analysis to make sure that the Kaplan–Meier estimator is well behaved. See \cite{he2013quantile},  \cite{zhou2017model}, \cite{chen2018robust}, and so on.
The assumption (C2) ensures that the $\omega_k$ of each active variable $X_k$ is not too close to 0, which is a common assumption in works on marginal screening (see, for instance, \cite{fan2016interaction}, \cite{chen2018robust}, \cite{guo2022stable}, \cite{wang2023spearman}, \cite{zhou2023bolt}).

\begin{theorem}[Sure screening property]
	\label{Sure screening property}
	Under condition (C1), there exists a positive constant $\Gamma >0$, such that
	\begin{align}
		P\left(\max _{1 \leqslant k \leqslant p}\left|\hat{\omega}_k-\omega_k\right| \geqslant c n^{-\kappa}\right) \leqslant O\left(p \exp \left\{-\Gamma n^{1-2 \kappa}\right\}\right) . \label{scp.1}
	\end{align}
	Under Condition (C2), we have
	\begin{align}
		P(\mathcal{A} \subseteq \hat{\mathcal{A}}) \geqslant 1-O\left(s \exp \left\{-\Gamma n^{1-2 \kappa}\right\}\right), \label{scp.2}
	\end{align}
	where $s=|\mathcal{A}|$ is the cardinality of $\mathcal{A}$. 
\end{theorem}
We can point out the size of the reduced model after XIM-SIS screening in the following theorem.
\begin{theorem}[Controlling false discoveries]
	\label{Controlling false discoveries}
	We have
	\begin{align}
		P\left(|\hat{\mathcal{A}}| \leqslant 2 c^{-1} n^\kappa \sum_{k=1}^p \omega_k\right) \geqslant 1-O\left(p \exp \left\{-\Gamma n^{1-2 \kappa}\right\}\right) . \label{cfd}
	\end{align}
\end{theorem}
This theorem suggests that after screening, the model tends to have a polynomial size with a probability close to 1 if $\sum_{k=1}^p \omega_k=O\left(n^a\right)$ for some $a>0$. Additionally, $Y$ is independent of $X_{\mathcal{A}^c}$ given $X_{\mathcal{A}}$. It implies a stronger dependency of $Y$ on $X_{\mathcal{A}}$ compared to $X_{\mathcal{A}^c}$. In simpler terms, $\omega_k$ for $k \in \mathcal{A}$ are larger than $\omega_k$ for $k \in \mathcal{A}^c$, formulated as:
\begin{itemize}
	\item[(C3)] $\min\limits _{k \in \mathcal{A}} \omega_k-\max \limits_{k \in \mathcal{A}^c} \omega_k=\delta>0$ for some constant $\delta$.
\end{itemize}
With assumption (C3), the XIM-SIS procedure exhibits the property of ranking consistency and the proof is given in the Appendix.
\begin{theorem}
	\label{Theorem 3}
	Under the conditions $(\mathrm{C} 1)$ and $(\mathrm{C} 3)$, we have
	\begin{align}
		P\left(\min _{k \in \mathcal{A}} \hat{\omega}_k>\max _{k \in \mathcal{A}^c} \hat{\omega}_k\right) \geqslant 1-O\left(p \exp \left\{-\Gamma n \delta^2\right\}\right) . \label{th3}
	\end{align}
\end{theorem}
Without necessary to assume that $\delta$ remains constant,  Theorem  \ref{Theorem 3}    implies that $\omega_k$ can consistently rank active variables above inactive with a probability approaching 1, given that $pe^{-\Gamma n\delta^2}=o(1)$. Instead, as $n \rightarrow \infty$, $\delta$ is permitted to tend towards zero. Specifically, the ranking consistency property is still true if we set $\delta=O\left(n^{-\gamma}\right), 0<\gamma<1 / 2$ and $\ln p=o\left(n^{1-2 \gamma}\right)$.

\section{Simulation Study}
This section assesses the finite-sample performance of our proposed method ``XIM-SIS'' by using numercal studies.  We  conduct a comparative analysis against existing competitors, such as the methods   FAST \citep{gorst2013independent}, CRIS  \citep{song2014censored},  CRSIS \citep{zhang2017correlation},  CSIRS
  \citep{zhou2017model},   rCDCS \citep{chen2018robust}, CCRIS  \citep{zhang2018censored},  rCFRS  \citep{hao2019robust},  CDCS  \citep{zhang2021model},  WSIRS \citep{song2022model},   and HSIC-SIS \citep{li2023model}.
 For the selection of $M$ in our method ``XIM-SIS'',  we choose the   $M=[\sqrt{n}]-1 $(XIM-SIS1), $[\sqrt{n}]$ (XIM-SIS2), $[\sqrt{n}]+1 $(XIN-SIS3) since $M$ is of the order of $\sqrt{n}$, where $[\cdot]$ is the least integer function.
  
  We assess the effectiveness of the screening methods using $N=500$  simulated replications. We employ three evaluation criteria as suggested by \cite{li2012feature} . 
  The first criterion is the  minimum model size that encompasses all active predictors, symbolized by 
  $\mathcal{S}$  serves as an indicator of the resulting model's complexity for each screening method. The closer the value of   $\mathcal{S}$
   is to the true  model size, the more effective the screening method is. We provide the 5\%, 25\%, 50\%(median), 75\%, 95\% quantiles , and  interquartile range (IQR) of  $\mathcal{S}$ from the 100 replications. 
   The second criterion is  the selection proportion of each true predictors   across the 100 replications,  denoted by  $\mathcal{P}_j$ , where $j$ is the index of truly important predictor.  The third criterion is  the selection proportion at which all active predictors are chosen simultaneously for a specific model  in the 100 replications, denoted by  $\mathcal{P}_a$. 
 An efficient screening method should yield a  $\mathcal{S}$ value close to the actual model size, and both  $\mathcal{P}_j$ and $\mathcal{P}_a$ should be close to 1. We determine the estimated model size to be $d_1=\lceil n / \log n\rceil$,  where $\lceil x\rceil$ denotes the integer part of $x$ and $d_2=n-1$. 
 All of results  related to Median,  IQR, and  $\mathcal{P}_j$ and $\mathcal{P}_a$  are listed in   Table \ref{ex:cox}-\ref{ex:log exp}.  Due to space constraints, the remaining results are included in the supplementary materials.

  \begin{example}\label{ex:cox}
  	Consider the survival time $T$ generated from the Cox proportional hazards model:
  	$$
  	h(t \mid X)=h_0(t) \exp \left(X^T \beta\right),
  	$$
  	where the baseline hazard function was set to be $h_0(t)=0.5$, $X \sim N_p(0, \boldsymbol{\Sigma})$ with $\boldsymbol{\Sigma}=$ $\left(0.6^{|i-j|}\right)$ for $i, j=1, \ldots, p$, $\beta=(0.35,0.35,0.35,0.35,0.35,0, \ldots, 0)^T$, $p=2000$, $n=100, 200$. Here $T$ was generated with formula $T=-\frac{\log (u)}{\lambda \exp \left(X^T \beta\right)}$, where $u \sim U(0,1)$ \citep{austin2012generating}.  The censoring time $C \sim \operatorname{Unif}(0, c)$ with $c$ chosen to achieve censoring rate $CR=30\%, 50\%$.
  \end{example}

  \begin{table}[t]
  	\setlength{\tabcolsep}{11pt} 
  	\centering
  	\caption{The median and IQR of $\mathcal{S}$, the proportion of $\mathcal{P}_j$ and $\mathcal{P}_a$ for Example \ref{ex:cox}}
  	\label{tab:cox}
  	\resizebox{\textwidth}{!}{ 
  		\begin{tabular}{lllrrllllllllllll} 
  			\toprule 
  			\multicolumn{1}{l}{\multirow{2}{*}{CR}}   & \multicolumn{1}{l}{\multirow{2}{*}{$n$}} & \multicolumn{1}{l}{\multirow{2}{*}{Method}} &  \multicolumn{1}{r}{\multirow{2}{*}{Median}}&\multicolumn{1}{r}{\multirow{2}{*}{IQR}}& \multicolumn{6}{l}{$d_1=n/log(n)$} & \multicolumn{6}{l}{$d_2=n-1$}  \\ \cmidrule(lr){6-11} \cmidrule(lr){12-17} 
  			& &  & &&$\mathcal{P}_1$ &$\mathcal{P}_2$ &$\mathcal{P}_3$ &$\mathcal{P}_4$   &$\mathcal{P}_5$& \multicolumn{1}{l}{$\mathcal{P}_a$}  &$\mathcal{P}_1$ &$\mathcal{P}_2$ &$\mathcal{P}_3$ & $\mathcal{P}_4$ &$\mathcal{P}_5$&$\mathcal{P}_a$ \\ \hline \addlinespace
  			\multirow[t]{24}{*}{30\%}& \multirow[t]{12}{*}{100}&  FAST & 6 & 4 &0.942 &0.992 &1.000 &0.988 &0.958 &0.884 &0.988 &1.000 &1.000 &0.998 &0.990 &0.978 
  			\\
  			& & CRIS  & 41 &126 &0.618 &0.838 & 0.896 & 0.824 &  0.630 &0.380 & 0.856 & 0.942 & 0.964 &0.934 & 0.832 &0.690 
  			\\
  			& & CCRIS  & 2000 &12 &0.006 & 0.024 & 0.024 & 0.020 &0.022 &0.000 & 0.026 & 0.052 & 0.048 & 0.046 & 0.030 &0.002 
  			\\
  			& & CSIRS  & 5 &4 & 0.934 & 0.992 & 0.998 & 0.994 &  0.952 &0.880 & 0.984 & 1.000 & 1.000 & 0.996 & 0.986 &0.970 
  			\\
  			& & WSIRS & 10 &30 & 0.786 & 0.914 & 0.938 & 0.918 &  0.818 &0.664 & 0.924 & 0.982 & 0.982 & 0.978 & 0.938 &0.880 
  			\\
  			& & rCFRS & 8 &14 &0.890 &0.974 &0.994 & 0.978 &0.896 &0.780 &0.964 &1.000 &0.998 & 0.994 &0.976 &0.934 
  			\\
  			& & rCDCS & 9 &21 & 0.844 & 0.958 &0.990 & 0.968 &0.870 &0.710 & 0.954 & 0.998 & 0.996 & 0.988 & 0.962 &0.904 
  			\\
  			& & CRSIS & 7 &10 & 0.914 & 0.984 & 0.992 & 0.986 &  0.900 &0.816 & 0.966 & 0.998 & 1.000 & 0.998 & 0.986 &0.948 
  			\\
  			& & CDCS & 6 &5 & 0.932 & 0.988 & 1.000 & 0.990 &  0.940 &0.854 & 0.984 & 1.000 & 1.000 & 0.998 & 0.990 &0.974 
  			\\
  			& &  HSIC-SIS & 20 &61 & 0.754 & 0.914 & 0.972 & 0.934 &0.766 &0.512 & 0.918 & 0.974 & 0.998 & 0.990 & 0.910 &0.810 
  			\\
  			& & XIM-SIS1  & 10 &19 &0.848 &0.960 &0.990 &0.974 & 0.878 &0.722 &0.956 &0.988 &1.000 &0.990 & 0.974 &0.914 
  			\\
  			& & XIM-SIS2  & 9 & 17 & 0.860 & 0.962 & 0.992 & 0.982 & 0.888 & 0.738 & 0.960 & 0.990 & 1.000 & 0.996 & 0.972 &0.920 
  			\\
  			& & XIM-SIS3  & 8 & 15 & 0.872 & 0.972 & 0.992 & 0.984 & 0.898 & 0.768 & 0.966 & 0.990 & 1.000 & 0.998 & 0.978 &0.934 
  			\\ \addlinespace
  			& \multirow[t]{12}{*}{200}& FAST  & 5 &0 & 1.000 & 1.000 & 1.000 & 1.000 &  1.000 &1.000 & 1.000 & 1.000 & 1.000 & 1.000 &   1.000 &1.000 
  			\\
  			& & CRIS & 5 &4 & 0.964 & 0.986 & 1.000 & 0.994 &  0.950 &0.914 & 0.986 & 1.000 & 1.000 & 1.000 &  0.986 &0.974 
  			\\
  			& & CCRIS & 2000 &0 & 0.012 & 0.024 & 0.026 & 0.016 &  0.012 &0.000 & 0.020 & 0.036 & 0.048 & 0.034 & 0.018 &0.006 
  			\\
  			& & CSIRS & 5 &0 & 1.000 & 1.000 & 1.000 & 1.000 &  1.000 &1.000 & 1.000 & 1.000 & 1.000 & 1.000 & 1.000 &1.000 
  			\\
  			& & WSIRS & 5 &1 & 0.938 & 0.960 &0.974 & 0.968 &0.926 &0.904 &0.978 & 0.990 & 0.992 & 0.994 &0.982 &0.974 
  			\\
  			& & rCFRS & 5 &0 & 1.000 & 1.000 & 1.000 & 1.000 &  0.998 &0.998 & 1.000 & 1.000 & 1.000 & 1.000 & 1.000 &1.000 
  			\\
  			& & rCDCS & 5 &0 & 0.998 & 1.000 & 1.000 & 1.000 &  0.996 &0.994 & 1.000 & 1.000 & 1.000 & 1.000 & 0.998 &0.998 
  			\\
  			& & CRSIS & 5 &0 & 1.000 & 1.000 & 1.000 & 1.000 &  0.998 &0.998 & 1.000 & 1.000 & 1.000 & 1.000 & 1.000 &1.000 
  			\\
  			& & CDCS & 5 &0 & 1.000 & 1.000 & 1.000 & 1.000 &  1.000 &1.000 & 1.000 & 1.000 & 1.000 & 1.000 & 1.000 &1.000 
  			\\
  			& &  HSIC-SIS & 5 &1 & 1.000 & 1.000 & 1.000 & 1.000 &  0.988 &0.988 & 1.000 & 1.000 & 1.000 & 1.000 & 1.000 &1.000 
  			\\
  			& & XIM-SIS1  & 5 &0 &0.996 &1.000 &1.000 &1.000 & 0.988 &0.984 &1.000 &1.000 &1.000 &1.000 & 0.998 &0.998 
  			\\
  			& & XIM-SIS2  & 5 & 0 & 0.998 & 1.000 & 1.000 & 1.000 & 0.988 & 0.986 & 1.000 & 1.000 & 1.000 & 1.000 & 0.996 &0.996 
  			\\
  			& & XIM-SIS3  & 5 & 0 & 0.998 & 1.000 & 1.000 & 1.000 & 0.994 & 0.992 & 1.000 & 1.000 & 1.000 & 1.000 & 0.998 &0.998 
  			\\ 
  			\hline 
  			\addlinespace 
  			\multirow[t]{24}{*}{50\%}& \multirow[t]{12}{*}{100}& FAST  & 8 &12 & 0.888 & 0.978 & 0.996 & 0.976 &  0.904 &0.780 & 0.970 & 1.000 & 1.000 & 0.994 &   0.978 &0.944 
  			\\
  			& & CRIS & 41 &115 & 0.646 &0.832 & 0.904 &0.836 &0.638 &0.372 &0.858 & 0.950 &0.966 &0.932 &0.854 &0.688 
  			\\
  			& & CCRIS & 2000 &2 & 0.000 & 0.000 & 0.002 & 0.000 &  0.000 &0.000 & 0.002 & 0.004 & 0.002 & 0.000 & 0.004 &0.000 
  			\\
  			& & CSIRS & 8 &17 & 0.872 & 0.976 & 0.992 & 0.976 &  0.876 &0.740 & 0.954 & 0.996 & 0.998 & 0.992 & 0.966 &0.912 
  			\\
  			& & WSIRS & 150 &375 & 0.354 & 0.520 &0.568 & 0.540 &0.386 &0.194 & 0.580 & 0.726 & 0.776 & 0.726 &0.594 &0.418 
  			\\
  			& & rCFRS & 40 &103 & 0.644 &0.852 & 0.906 &0.866 &0.672 &0.374 & 0.858 & 0.944 & 0.976 & 0.950 &0.880 &0.720 
  			\\
  			& & rCDCS & 53 &150 & 0.604 & 0.818 &0.880 & 0.852 &0.624 &0.304 & 0.828 & 0.938 & 0.968 &0.944 &0.842 &0.654 
  			\\
  			& & CRSIS & 34 &93 & 0.656 & 0.860 &0.912 & 0.894 &0.694 &0.398 & 0.870 & 0.956 & 0.978 & 0.964 & 0.886 &0.742 
  			\\
  			& & CDCS & 7 &17 & 0.876 & 0.972 & 0.996 & 0.980 &  0.874 &0.748 & 0.962 & 0.998 & 1.000 & 0.990 & 0.972 &0.926 
  			\\
  			& &  HSIC-SIS & 29 &108 & 0.694 & 0.888 & 0.948 & 0.908 &0.718 &0.424 & 0.876 & 0.968 & 0.990 & 0.968 & 0.844 &0.716 
  			\\
  			& & XIM-SIS1  & 60 &162 &0.594 &0.808 &0.874 &0.832 & 0.604 &0.286 &0.790 &0.924 &0.958 &0.950 & 0.834 &0.608 
  			\\
  			& & XIM-SIS2  & 49 & 140 & 0.614 & 0.822 & 0.880 & 0.850 & 0.622 & 0.308 & 0.820 & 0.932 & 0.964 & 0.952 & 0.842 &0.640 
  			\\
  			& & XIM-SIS3  & 44 & 126 & 0.626 & 0.838 & 0.892 & 0.862 & 0.646 & 0.338 & 0.836 & 0.934 & 0.968 & 0.966 & 0.856 &0.666 
  			\\
  			\addlinespace
  			& \multirow[t]{3}{*}{200}& FAST  & 5 &0 & 1.000 & 1.000 & 1.000 & 1.000 &  0.998 &0.998 & 1.000 & 1.000 & 1.000 & 1.000 &   1.000 &1.000 
  			\\
  			& & CRIS & 6 &4 & 0.948 & 0.996 & 0.992 &0.992 &  0.944 &0.892 & 0.986 & 0.998 & 0.998 & 0.998 &  0.982 &0.966 
  			\\
  			& & CCRIS & 2000 &0 & 0.000 & 0.000 & 0.000 & 0.000 &  0.000 &0.000 & 0.000 & 0.000 & 0.000 & 0.000 & 0.000 &0.000 
  			\\
  			& & CSIRS & 5 &0 & 0.998 & 1.000 & 1.000 & 1.000 &  1.000 &0.998 & 1.000 & 1.000 & 1.000 & 1.000 & 1.000 &1.000 
  			\\
  			& & WSIRS & 28 &150 & 0.674 & 0.756 & 0.794 & 0.762 &0.630 &0.536 & 0.840 & 0.874 & 0.874 & 0.858 &0.832 &0.780 
  			\\
  			& & rCFRS & 5 &3 & 0.970 & 1.000 & 1.000 & 1.000 &  0.970 &0.940 & 0.996 & 1.000 & 1.000 & 1.000 & 0.998 &0.994 
  			\\
  			& & rCDCS & 6 &5 & 0.958 & 0.998 & 1.000 & 1.000 &  0.968 &0.926 & 0.996 & 1.000 & 1.000 & 1.000 & 0.992 &0.988 
  			\\
  			& & CRSIS & 5 &2 & 0.984 & 1.000 & 1.000 & 1.000 &  0.984 &0.968 & 0.998 & 1.000 & 1.000 & 1.000 & 0.998 &0.996 
  			\\
  			& & CDCS & 5 &0 & 1.000 & 1.000 & 1.000 & 1.000 &  1.000 &1.000 & 1.000 & 1.000 & 1.000 & 1.000 & 1.000 &1.000 
  			\\
  			& &  HSIC-SIS & 5 &2 & 0.988 & 0.998 & 1.000 & 1.000 &  0.984 &0.970 & 1.000 & 1.000 & 1.000 & 1.000 & 1.000 &1.000 
  			\\
  			& & XIM-SIS1  & 7 &9 &0.932 &0.994 &1.000 &0.992 & 0.920 &0.860 &0.988 &1.000 &1.000 &1.000 & 0.986 &0.974 
  			\\
  			& & XIM-SIS2  & 6 & 8 & 0.936 & 0.994 & 0.998 & 0.996 & 0.940 & 0.884 & 0.994 & 1.000 & 1.000 & 1.000 & 0.992 &0.986 
  			\\
  			& & XIM-SIS3  & 6 & 7 & 0.938 & 1.000 & 0.998 & 0.996 & 0.948 & 0.892 & 0.994 & 1.000 & 1.000 & 1.000 & 0.986 &0.980 \\ 
  			\bottomrule  
  \end{tabular} }
  \end{table}

  \begin{table}[htbp]
	\setlength{\tabcolsep}{26pt} 
	\centering
	\caption{The quantiles of the minimum model size $S$ for Example \ref{ex:cox}}
	\label{tab: cox size}
	\resizebox{\textwidth}{!}{ 
		\begin{tabular}{llrrrrrrrrrr} 
  \toprule 
			\multicolumn{1}{l}{\multirow{2}{*}{CR}}    & \multicolumn{1}{l}{\multirow{2}{*}{Method}} & \multicolumn{5}{l}{$n=100$} & \multicolumn{5}{l}{$n=200$} \\ \cmidrule(lr){3-7} \cmidrule(lr){8-12}
   &   & 5\% & 25\% & 50\% & 75\% & \multicolumn{1}{l}{95\%}   & 5\% & 25\% & 50\% & 75\% & 95\%  \\ \hline \addlinespace 
			\multirow[t]{12}{*}{30\%}& FAST& 5 & 5 & 6 & 9 & 45 & 5 & 5 & 5 & 5 &  6 
\\
			& CRIS& 6 &  12 & 41 & 138 & 616 & 5 & 5 & 5 &9 & 65 
\\
			& CCRIS& 1256 & 1988 & 2000 & 2000 & 2000 & 1760 & 2000 & 2000 & 2000 &2000 
\\
			& CSIRS& 5 & 5 & 5 & 9 & 52 & 5 & 5 & 5 & 5 &5 
\\
			& WSIRS& 5 & 6 & 10 & 36 & 284 & 5 & 5 & 5 & 6 &97 
\\
			& rCFRS& 5 & 5 & 8 & 19 & 138 & 5 & 5 & 5 & 5 &7 
\\
			& rCDCS& 5 & 6 & 9 & 27 & 180 & 5 & 5 & 5 & 5 &8 
\\
			& CRSIS& 5 & 5 & 7 & 15 & 102 & 5 & 5 & 5 & 5 &6 
\\
			& CDCS& 5 & 5 & 6 & 10 & 62 & 5 & 5 & 5 & 5 &6 
\\
			& HSIC-SIS& 5 & 8 & 20 & 69 & 366 & 5 & 5 & 5 & 6 &13 
\\
			& XIM-SIS1   &5 &6 &10 &25 &175 &5 &5 &5 &5 &12 
\\
 & XIM-SIS2   & 5 & 6 & 9 & 23 & 146 & 5 & 5 & 5 & 5 &9 
\\
 & XIM-SIS3   & 5 & 5 & 8 & 20 & 130 & 5 & 5 & 5 & 5 &8 \\
 \addlinespace 
		\multirow[t]{12}{*}{50\%}	& FAST& 5 & 5 & 8 & 17 & 102 & 5 & 5 & 5 & 5 &  6 
\\
			& CRIS& 6 & 12 & 41 & 127 & 546 & 5 & 5 & 6 & 9 & 139 
\\
			& CCRIS& 1973 & 1998 & 2000 & 2000 & 2000 & 2000 & 2000 & 2000 & 2000 &2000 
\\
			& CSIRS& 5 & 5 & 8 & 22 & 173 & 5 & 5 & 5 & 5 &7 
\\
			& WSIRS& 6 & 32 & 150 & 407 & 1133 & 5 & 8 & 28 & 158 &955 
\\
			& rCFRS& 5 & 13 & 40 & 116 & 618 & 5 & 5 & 5 & 8 &42 
\\
			& rCDCS& 6 & 16 & 53 & 166 & 836 & 5 & 5 & 6 & 10 &54 
\\
			& CRSIS& 5 & 10 & 34 & 103 & 439 & 5 & 5 & 5 & 7 &27 
\\
			& CDCS& 5 & 5 & 7 & 22 & 150 & 5 & 5 & 5 & 5 &6 
\\
			& HSIC-SIS& 5 & 9 & 29 & 117 & 594 & 5 & 5 & 5 & 7 &23 
\\
			& XIM-SIS1   &6 &18 &60 &180 &631 &5 &5 &7 &14 &117 
\\
 & XIM-SIS2   & 5 & 15 & 49 & 155 & 554 & 5 & 5 & 6 & 13 &101 
\\
 & XIM-SIS3   & 5 & 15 & 44 & 141 & 517 & 5 & 5 & 6 & 12 &95 \\ 
			\bottomrule  
	\end{tabular}}
	\parbox{\linewidth}{\footnotesize 
    }
\end{table}

  The simulation results, as shown in Tables \ref{tab:cox} and  \ref{tab: cox size}, indicate that when the chosen model size is $d_1=\lceil n / \log n\rceil$, the method CSIRS, FAST, and CDCS exhibit slightly superior performance compared to other methods when the sample size is 100, under both 30\% and 50\% censoring rates. However, this slight advantage diminishes as the sample size increases from 100 to 200. 
  Expect for CCRIS and WSIRS, the remaining methods demonstrate equally effective performance when the sample size is set to 200.    They both achieve a high  accuracy  for $\mathcal{P}_a$.
  Furthermore, as the censoring rate escalates from 30\% to 50\%, the efficiency of all methods  deteriorates.
  When  $d_2=n-1$, we have similar conclusions.  In general, our proposed XIM-SIS method could achieve comparable performance to existing screening methods for the traditional Cox's model at moderate sample sizes.

\begin{example}\label{ex:transformation} 
	Consider the survival time $T$ generated from the linear hazard model:   
	$$
	H(T)=-\beta^{\top} \boldsymbol{X}+\epsilon,
	$$
	where $H(t)=\log [0.5\{\exp (2 t)-1\}], \beta=(-1,-0.9,0,0,0,0,0,0,0.8,1, 0,\ldots, 0)^{\top}, p=2000$, $n=100, 200$, $s=4$. The ultra-high dimensional covariates  $X \sim N_p(0, \boldsymbol{\Sigma})$ with $\boldsymbol{\Sigma}=$ $\left(0.5^{|i-j|}\right)$ for $i, j=1, \ldots, p$. The error term $\epsilon$ is student's t distribution with degree of freedom 1. The censoring time $C \sim \operatorname{Unif}(0, c)$ with $c$ chosen to achieve censoring rate $CR=20\%, 40\%$.
\end{example}

\begin{table}[!htbp]
	\setlength{\tabcolsep}{16pt} 
	\centering
	\caption{The median and IQR of $\mathcal{S}$, the proportion of $\mathcal{P}_j$ and $\mathcal{P}_a$ for Example \ref{ex:transformation}}
	\label{tab:transformation}
	\resizebox{\textwidth}{!}{ 
		\begin{tabular}{lllrrllllllllll} 
			\toprule 
			\multicolumn{1}{l}{\multirow{2}{*}{CR}}   & \multicolumn{1}{l}{\multirow{2}{*}{$n$}} & \multicolumn{1}{l}{\multirow{2}{*}{Method}} &  \multicolumn{1}{r}{\multirow{2}{*}{Median}}&\multicolumn{1}{r}{\multirow{2}{*}{IQR}}& \multicolumn{5}{l}{$d_1=n/log(n)$} & \multicolumn{5}{l}{$d_2=n-1$} \\ 
			\cmidrule(lr){6-10} \cmidrule(lr){11-15}& &   & && $\mathcal{P}_1$ & $\mathcal{P}_2$ & $\mathcal{P}_9$ & $\mathcal{P}_{10}$   & \multicolumn{1}{l}{$\mathcal{P}_a$} & $\mathcal{P}_1$ & $\mathcal{P}_2$ & $\mathcal{P}_9$ & $\mathcal{P}_{10}$ &  $\mathcal{P}_a$ \\ \hline \addlinespace 
			\multirow[t]{24}{*}{20\%}& \multirow[t]{12}{*}{200}& FAST& 4 &3 & 0.994 & 0.998 & 0.972 & 0.988 & 0.958 & 1.000 & 1.000 & 0.998 & 0.998 &  0.996 
			\\
			& & CRIS& 34 &179 & 0.828 &  0.792 & 0.768 & 0.854 & 0.514 & 0.908 & 0.890 & 0.910 &0.956 & 0.758 
			\\
			& & CCRIS& 908 &1385 & 0.962 & 0.956 & 0.170 & 0.228 & 0.104 & 0.994 & 0.984 & 0.284 & 0.392 &0.218 
			\\
			& & CSIRS& 4 &0 & 1.000 & 1.000 & 0.990 & 0.998 & 0.988 & 1.000 & 1.000 & 0.998 & 1.000 &0.998 
			\\
			& & WSIRS& 4 &0 & 1.000 & 1.000 & 0.990 & 0.998 & 0.988 & 1.000 & 1.000 & 0.998 & 1.000 &0.998 
			\\
			& & rCFRS& 4 &1 & 1.000 & 0.996 & 0.982 & 0.996 & 0.974 & 1.000 & 1.000 & 0.996 & 1.000 &0.996 
			\\
			& & rCDCS& 4 &1 & 1.000 & 0.996 & 0.978 & 0.996 & 0.970 & 1.000 & 1.000 & 0.994 & 1.000 &0.994 
			\\
			& & CRSIS& 4 &1 & 1.000 & 0.996 & 0.984 & 0.996 & 0.976 & 1.000 & 1.000 & 0.998 & 1.000 &0.998 
			\\
			& & CDCS& 4 &1 & 0.998 & 0.998 & 0.978 & 0.996 & 0.970 & 1.000 & 1.000 & 0.996 & 1.000 &0.996 
			\\
			& &  HSIC-SIS& 5 &7 & 0.990 & 0.992 & 0.950 & 0.992 & 0.926 & 1.000 & 0.998 & 0.982 & 1.000 &0.980 
			\\
			& & XIM-SIS1   & 41 &131 &0.718 &0.654 &0.958 &0.984 &0.480 &0.902 &0.872 &0.990 &1.000 &0.800 
			\\
			& & XIM-SIS2  & 47 & 155 & 0.696 & 0.618 & 0.958 & 0.984 & 0.454 & 0.888 & 0.860 & 0.990 & 1.000 &0.784 
			\\
			& & XIM-SIS3  & 65 & 200 & 0.654 & 0.570 & 0.962 & 0.984 & 0.406 & 0.852 & 0.836 & 0.992 & 1.000 &0.740 
			\\ \addlinespace 
			& \multirow[t]{12}{*}{300}& FAST& 4 &0 & 1.000 & 1.000 & 0.998 & 1.000 & 0.998 & 1.000 & 1.000 & 1.000 & 1.000 &  1.000 
			\\
			& & CRIS& 7 &28 & 0.946 & 0.928 & 0.914 & 0.960 & 0.802 & 0.974 & 0.962 & 0.978 & 0.992 & 0.926 
			\\
			& & CCRIS& 956 &1496 & 0.998 & 0.986 & 0.192 & 0.262 & 0.134 & 1.000 & 0.996 & 0.382 & 0.458 &0.290 
			\\
			& & CSIRS& 4 &0 & 1.000 & 1.000 & 1.000 & 1.000 & 1.000 & 1.000 & 1.000 & 1.000 & 1.000 &1.000 
			\\
			& & WSIRS& 4 &0 & 1.000 & 1.000 & 1.000 & 1.000 & 1.000 & 1.000 & 1.000 & 1.000 & 1.000 &1.000 
			\\
			& & rCFRS& 4 &0 & 1.000 & 1.000 & 1.000 & 1.000 & 1.000 & 1.000 & 1.000 & 1.000 & 1.000 &1.000 
			\\
			& & rCDCS& 4 &0 & 1.000 & 1.000 & 1.000 & 1.000 & 1.000 & 1.000 & 1.000 & 1.000 & 1.000 &1.000 
			\\
			& & CRSIS& 4 &0 & 1.000 & 1.000 & 1.000 & 1.000 & 1.000 & 1.000 & 1.000 & 1.000 & 1.000 &1.000 
			\\
			& & CDCS& 4 &0 & 1.000 & 1.000 & 1.000 & 1.000 & 1.000 & 1.000 & 1.000 & 1.000 & 1.000 &1.000 
			\\
			& &  HSIC-SIS& 4 &0 & 1.000 & 1.000 & 0.998 & 1.000 & 0.998 & 1.000 & 1.000 & 1.000 & 1.000 &1.000 
			\\
			& & XIM-SIS1   & 6 &10 &0.950 &0.946 &1.000 &1.000 &0.904 &0.994 &0.988 &1.000 &1.000 &0.984 
			\\
			& & XIM-SIS2  & 6 & 13 & 0.946 & 0.944 & 1.000 & 1.000 & 0.896 & 0.994 & 0.988 & 1.000 & 1.000 &0.984 
			\\
			& & XIM-SIS3  & 6 & 17 & 0.942 & 0.922 & 1.000 & 1.000 & 0.874 & 0.990 & 0.988 & 1.000 & 1.000 &0.980 
			\\ \hline \addlinespace 
			\multirow[t]{24}{*}{40\%}& \multirow[t]{12}{*}{200}& FAST& 4 &1 & 0.998 & 0.998 & 0.988 & 0.996 & 0.980 & 1.000 & 0.998 & 1.000 & 1.000 &  0.998 
			\\
			& & CRIS& 93 &373 & 0.642 & 0.598 & 0.788 &0.878 & 0.348 & 0.788 & 0.756 & 0.934 & 0.972 & 0.644 
			\\
			& & CCRIS& 2000 &2 & 0.832 & 0.814 & 0.002 & 0.002 & 0.000 & 0.974 & 0.950 & 0.006 & 0.012 &0.000 
			\\
			& & CSIRS& 4 &1 & 1.000 & 0.998 & 0.984 & 0.998 & 0.980 & 1.000 & 1.000 & 0.998 & 1.000 &0.998 
			\\
			& & WSIRS& 56 &199 & 0.746 & 0.684 & 0.566 & 0.666 & 0.440 & 0.916 & 0.868 & 0.836 & 0.874 &0.734 
			\\
			& & rCFRS& 5 &3 & 1.000 & 0.994 & 0.950 & 0.984 & 0.932 & 1.000 & 0.996 & 0.988 & 0.996 &0.982 
			\\
			& & rCDCS& 5 &4 & 0.994 & 0.984 & 0.946 & 0.982 & 0.908 & 1.000 & 0.996 & 0.976 & 0.996 &0.968 
			\\
			& & CRSIS& 4 &3 & 1.000 & 0.994 & 0.960 & 0.986 & 0.944 & 1.000 & 0.996 & 0.992 & 0.996 &0.986 
			\\
			& & CDCS& 4 &1 & 1.000 & 0.998 & 0.986 & 0.996 & 0.980 & 1.000 & 1.000 & 0.996 & 1.000 &0.996 
			\\
			& &  HSIC-SIS& 6 &8 & 0.992 & 0.986 & 0.946 & 0.982 & 0.908 & 1.000 & 0.998 & 0.986 & 0.996 &0.980 
			\\
			& & XIM-SIS1   & 209 &526 &0.408 &0.370 &0.910 &0.966 &0.180 &0.706 &0.670 &0.972 &0.996 &0.488 
			\\
			& & XIM-SIS2  & 260 & 592 & 0.374 & 0.346 & 0.918 & 0.972 & 0.162 & 0.658 & 0.616 & 0.976 & 0.996 &0.432 
			\\
			& & XIM-SIS3  & 346 & 638 & 0.336 & 0.300 & 0.930 & 0.972 & 0.138 & 0.618 & 0.558 & 0.982 & 0.998 &0.378 
			\\ \addlinespace 
			& \multirow[t]{12}{*}{300}& FAST& 4 &0 & 1.000 & 1.000 & 1.000 & 1.000 & 1.000 & 1.000 & 1.000 & 1.000 & 1.000 &  1.000 
			\\
			& & CRIS& 14 &102 & 0.792 & 0.762 & 0.954 &0.982 & 0.670 & 0.902 & 0.896 & 0.986 & 0.994 & 0.846 
			\\
			& & CCRIS& 2000 &0 & 0.936 & 0.942 & 0.004 & 0.004 & 0.000 & 0.996 & 0.988 & 0.010 & 0.014 &0.002 
			\\
			& & CSIRS& 4 &0 & 1.000 & 1.000 & 1.000 & 1.000 & 1.000 & 1.000 & 1.000 & 1.000 & 1.000 &1.000 
			\\
			& & WSIRS& 16 &108 & 0.802 & 0.790 & 0.734 & 0.772 & 0.660 & 0.950 & 0.944 & 0.926 & 0.928 &0.880 
			\\
			& & rCFRS& 4 &0 & 1.000 & 0.998 & 1.000 & 0.998 & 0.996 & 1.000 & 1.000 & 1.000 & 1.000 &1.000 
			\\
			& & rCDCS& 4 &0 & 1.000 & 0.998 & 1.000 & 1.000 & 0.998 & 1.000 & 1.000 & 1.000 & 1.000 &1.000 
			\\
			& & CRSIS& 4 &0 & 1.000 & 1.000 & 0.998 & 1.000 & 0.998 & 1.000 & 1.000 & 1.000 & 1.000 &1.000 
			\\
			& & CDCS& 4 &0 & 1.000 & 1.000 & 1.000 & 1.000 & 1.000 & 1.000 & 1.000 & 1.000 & 1.000 &1.000 
			\\
			& &  HSIC-SIS& 4 &0 & 1.000 & 1.000 & 1.000 & 0.998 & 0.998 & 1.000 & 1.000 & 1.000 & 1.000 &1.000 
			\\
			& & XIM-SIS1   & 24 &118 &0.792 &0.740 &0.992 &0.998 &0.624 &0.934 &0.914 &1.000 &1.000 &0.862 
			\\
			& & XIM-SIS2  & 29 & 155 & 0.770 & 0.718 & 0.996 & 0.998 & 0.606 & 0.928 & 0.898 & 1.000 & 1.000 &0.838 
			\\
			& & XIM-SIS3  & 35 & 191 & 0.738 & 0.690 & 0.996 & 0.998 & 0.564 & 0.916 & 0.872 & 1.000 & 1.000 &0.810 \\
			\bottomrule  
	\end{tabular}}
\end{table}

\begin{table}[!htbp]
	\setlength{\tabcolsep}{26pt} 
	\centering
	\caption{The quantiles of the minimum model size $S$ for Example \ref{ex:transformation}}
	\label{tab:transformation size}
	\resizebox{\textwidth}{!}{ 
		\begin{tabular}{llrrrrrrrrrr} 
  \toprule 
			\multicolumn{1}{l}{\multirow{2}{*}{CR}}    & \multicolumn{1}{l}{\multirow{2}{*}{Method}} & \multicolumn{5}{l}{$n=200$} & \multicolumn{5}{l}{$n=300$} \\ \cmidrule(lr){3-7} \cmidrule(lr){8-12}
   &   & 5\% & 25\% & 50\% & 75\% & \multicolumn{1}{l}{95\%}   & 5\% & 25\% & 50\% & 75\% & 95\%  \\ \hline \addlinespace 
			\multirow[t]{12}{*}{20\%}& FAST& 4 & 4 & 4 & 7 & 33 & 4 & 4 & 4 & 4 &  6 
\\
			& CRIS& 4 &  10 & 34 & 189 & 1150 & 4 & 4 & 7 &32 & 435 
\\
			& CCRIS& 10 & 260 & 908 & 1645 & 1985 & 10 & 217 & 956 & 1713 &1986 
\\
			& CSIRS& 4 & 4 & 4 & 4 & 8 & 4 & 4 & 4 & 4 &4 
\\
			& WSIRS& 4 & 4 & 4 & 4 & 9 & 4 & 4 & 4 & 4 &4 
\\
			& rCFRS& 4 & 4 & 4 & 5 & 15 & 4 & 4 & 4 & 4 &5 
\\
			& rCDCS& 4 & 4 & 4 & 5 & 22 & 4 & 4 & 4 & 4 &5 
\\
			& CRSIS& 4 & 4 & 4 & 5 & 13 & 4 & 4 & 4 & 4 &5 
\\
			& CDCS& 4 & 4 & 4 & 5 & 17 & 4 & 4 & 4 & 4 &5 
\\
			& HSIC-SIS & 4 & 4 & 5 & 11 & 73 & 4 & 4 & 4 & 4 &10 
\\
			& XIM-SIS1   &4 &11 &41 &142 &697 &4 &4 &6 &14 &103 
\\
 & XIM-SIS2   & 4 & 13 & 47 & 168 & 697 & 4 & 4 & 6 & 17 &119 
\\
 & XIM-SIS3   & 4 & 15 & 65 & 215 & 827 & 4 & 4 & 6 & 21 &140 \\ 
 \addlinespace 
		\multirow[t]{12}{*}{40\%}	& FAST& 4 & 4 & 4 & 5 & 17 & 4 & 4 & 4 & 4 &  5 
\\
			& CRIS& 5 & 20 & 93 & 393 & 1828 & 4 & 5 & 14 & 107 & 1224 
\\
			& CCRIS& 1747 & 1998 & 2000 & 2000 & 2000 & 1957 & 2000 & 2000 & 2000 &2000 
\\
			& CSIRS& 4 & 4 & 4 & 5 & 14 & 4 & 4 & 4 & 4 &5 
\\
			& WSIRS& 4 & 14 & 56 & 213 & 658 & 4 & 5 & 16 & 113 &649 
\\
			& rCFRS& 4 & 4 & 5 & 7 & 55 & 4 & 4 & 4 & 4 &7 
\\
			& rCDCS& 4 & 4 & 5 & 8 & 82 & 4 & 4 & 4 & 4 &8 
\\
			& CRSIS& 4 & 4 & 4 & 7 & 41 & 4 & 4 & 4 & 4 &7 
\\
			& CDCS& 4 & 4 & 4 & 5 & 13 & 4 & 4 & 4 & 4 &4 
\\
			& HSIC-SIS& 4 & 4 & 6 & 12 & 59 & 4 & 4 & 4 & 4 &10 
\\
			& XIM-SIS1   &9 &65 &209 &591 &1290 &4 &7 &24 &125 &651 
\\
 & XIM-SIS2   & 11 & 80 & 260 & 671 & 1292 & 4 & 8 & 29 & 163 &707 
\\
 & XIM-SIS3   & 13 & 98 & 346 & 735 & 1467 & 4 & 9 & 35 & 200 &800 \\ 
			\bottomrule  
	\end{tabular}}
	\parbox{\linewidth}{\footnotesize 
    }
\end{table}

The simulation results are summarized in Table \ref{tab:transformation} and  \ref{tab:transformation size}. While the proposed XIM-SIS demonstrates superior performance over the model-free methods CCRIS and CRIS only when CR=20\% with $n=200$, the accuracy proportion $\mathcal{P}_a$   achieves consistently high levels as the sample size increases to $n=300$, regardless of the selected model size.
Furthermore, when the model size is set to $n-1$,  our proposed XIM-SIS demonstrates satisfactory performance, particularly at the sample size  $n=300$. Overall, the XIM-SIS method achieves performance comparable to existing screening methods for linear hazard models under moderate sample sizes.

\begin{example}\label{ex:aft}
	Consider the survival time $T$ generated  from the accelerated failure time model, which is another popular semiparametric model in survival analysis:
	$$
	\log (T)=X_1+0.8 X_2+X_7^2+\varepsilon
	$$
	where $\varepsilon \sim N(0, 1)$,  $X$ follows the same distribution as the covariates in Example \ref{ex:cox}. Same as Example \ref{ex:cox}, still set $p=2000$, $n=200$ and $300$, $CR=30\%, 50\%$.
\end{example}

\begin{table}[!htbp]
	\setlength{\tabcolsep}{20pt} 
	\centering
	\caption{The median and IQR of $\mathcal{S}$, the proportion of $\mathcal{P}_j$ and $\mathcal{P}_a$ for Example \ref{ex:aft}}
	\label{tab:aft}
	\resizebox{\textwidth}{!}{ 
		\begin{tabular}{lllrrllllllll}  \toprule 
			\multicolumn{1}{l}{\multirow{2}{*}{CR}}   & \multicolumn{1}{l}{\multirow{2}{*}{$n$}} & \multicolumn{1}{l}{\multirow{2}{*}{Method}} &  \multicolumn{1}{r}{\multirow{2}{*}{Median}}&\multicolumn{1}{r}{\multirow{2}{*}{IQR}} &  \multicolumn{4}{l}{$d_1=n/log(n)$} & \multicolumn{4}{l}{$d_2=n-1$} \\  \cmidrule(lr){6-9} \cmidrule(lr){10-13}
			& &  & &&$\mathcal{P}_1$ &$\mathcal{P}_2$ &$\mathcal{P}_7$ & \multicolumn{1}{l}{$\mathcal{P}_a$}   &$\mathcal{P}_1$ &$\mathcal{P}_2$ &$\mathcal{P}_7$ & $\mathcal{P}_a$ \\ \hline \addlinespace
			\multirow[t]{24}{*}{30\%}&\multirow[t]{12}{*}{200}& FAST & 754 &1127 &1.000 &1.000 &0.054 &0.054 &1.000 &1.000 &0.198 &0.198 
			\\
			& & CRIS  & 1067 &987 &0.960 &0.944 & 0.032 & 0.032 & 0.976 & 0.968 & 0.114 &0.114 
			\\
			& & CCRIS  & 147 &398 &0.996 & 0.996 & 0.234 & 0.228 & 1.000 & 0.996 & 0.576 & 0.572 
			\\
			& & CSIRS  & 1011 &1255 & 1.000 & 1.000 & 0.030 & 0.030 & 1.000 & 1.000 & 0.126 & 0.126 
			\\
			& & WSIRS & 879 &1134 & 0.982 & 0.974 & 0.022 & 0.022 & 0.992 & 0.990 & 0.146 & 0.146 
			\\
			& & rCFRS & 871 &1007 &1.000 &1.000 &0.032 & 0.032 &1.000 &1.000 &0.144 & 0.144 
			\\
			& & rCDCS & 19 &37 & 1.000 & 1.000 &0.700 & 0.700 & 1.000 & 1.000 & 0.982 & 0.982 
			\\
			& & CRSIS & 865 &1082 & 1.000 & 1.000 & 0.020 & 0.020 & 1.000 & 1.000 & 0.134 & 0.134 
			\\
			& & CDCS & 4 &3 & 1.000 & 1.000 & 0.984 & 0.984 & 1.000 & 1.000 & 1.000 & 1.000 
			\\
			& &  HSIC-SIS & 3 &0 & 1.000 & 1.000 & 0.996 & 0.996 & 1.000 & 1.000 & 1.000 & 1.000 
			\\
			& & XIM-SIS1  & 4 &14 &1.000 &0.998 &0.842 &0.840 &1.000 &0.998 &0.954 &0.952 
			\\
			& & XIM-SIS2  & 5 & 16 & 1.000 & 0.992 & 0.834 & 0.826 & 1.000 & 0.998 & 0.958 &0.956 
			\\
			& & XIM-SIS3  & 5 & 18 & 1.000 & 0.982 & 0.816 & 0.802 & 1.000 & 0.998 & 0.956 &0.954 
			\\ 
			\addlinespace
			& \multirow[t]{12}{*}{300}& FAST  & 750 &1011 & 1.000 & 1.000 & 0.106 & 0.106 & 1.000 & 1.000 & 0.282 & 0.282 
			\\
			& & CRIS & 1056 &974 & 0.994 & 0.992 & 0.038 & 0.038 & 0.996 & 0.996 & 0.150 & 0.150 
			\\
			& & CCRIS & 75 &235 & 0.996 & 0.996 & 0.442 & 0.442 & 0.998 & 0.998 & 0.784 & 0.782 
			\\
			& & CSIRS & 970 &1110 & 1.000 & 1.000 & 0.064 & 0.064 & 1.000 & 1.000 & 0.216 & 0.216 
			\\
			& & WSIRS & 824 &1007 & 0.966 & 0.968 &0.056 & 0.054 &0.984 & 0.988 & 0.234 & 0.226 
			\\
			& & rCFRS & 889 &1071 & 1.000 & 1.000 & 0.056 & 0.056 & 1.000 & 1.000 & 0.210 & 0.210 
			\\
			& & rCDCS & 6 &6 & 1.000 & 1.000 & 0.976 & 0.976 & 1.000 & 1.000 & 1.000 & 1.000 
			\\
			& & CRSIS & 928 &1067 & 1.000 & 1.000 & 0.048 & 0.048 & 1.000 & 1.000 & 0.192 & 0.192 
			\\
			& & CDCS & 4 &1 & 1.000 & 1.000 & 0.998 & 0.998 & 1.000 & 1.000 & 1.000 & 1.000 
			\\
			& &  HSIC-SIS & 3 &0 & 1.000 & 1.000 & 1.000 & 1.000 & 1.000 & 1.000 & 1.000 & 1.000 
			\\
			& & XIM-SIS1  & 3 &0 &1.000 &1.000 &0.992 &0.992 &1.000 &1.000 &1.000 &1.000 
			\\
			& & XIM-SIS2  & 3 & 0 & 1.000 & 1.000 & 0.984 & 0.984 & 1.000 & 1.000 & 1.000 &1.000 
			\\
			& & XIM-SIS3  & 3 & 0 & 1.000 & 1.000 & 0.982 & 0.982 & 1.000 & 1.000 & 1.000 &1.000 
			\\
			\hline 
			\addlinespace 
			\multirow[t]{24}{*}{50\%}&\multirow[t]{12}{*}{200}& FAST  & 911 &1044 & 1.000 & 1.000 & 0.022 & 0.022 & 1.000 & 1.000 & 0.132 & 0.132 
			\\
			& & CRIS & 1349 &857 & 0.646 &0.572 & 0.014 &0.010 &0.740 & 0.706 &0.072 &0.042 
			\\
			& & CCRIS & 507 &698 & 0.944 & 0.920 & 0.072 & 0.064 & 0.982 & 0.968 & 0.272 & 0.262 
			\\
			& & CSIRS & 1238 &1164 & 1.000 & 1.000 & 0.020 & 0.020 & 1.000 & 1.000 & 0.086 & 0.086 
			\\
			& & WSIRS & 1197 &971 & 0.862 & 0.836 &0.012 & 0.012 & 0.924 & 0.928 & 0.090 & 0.078 
			\\
			& & rCFRS & 941 &1047 & 1.000 &1.000 & 0.024 &0.024 & 1.000 & 1.000 & 0.126 & 0.126 
			\\
			& & rCDCS & 66 &127 & 1.000 & 1.000 &0.344 & 0.344 & 1.000 & 1.000 & 0.808 &0.808 
			\\
			& & CRSIS & 949 &1031 & 1.000 & 1.000 &0.018 & 0.018 & 1.000 & 1.000 & 0.130 & 0.130 
			\\
			& & CDCS & 7 &11 & 1.000 & 1.000 & 0.908 & 0.908 & 1.000 & 1.000 & 0.994 & 0.994 
			\\
			& &  HSIC-SIS & 3 &1 & 1.000 & 1.000 & 0.984 & 0.984 & 1.000 & 1.000 & 1.000 & 1.000 
			\\
			& & XIM-SIS1  & 57 &244 &0.972 &0.926 &0.460 &0.434 &0.996 &0.984 &0.724 &0.710 
			\\
			& & XIM-SIS2  & 64 & 244 & 0.968 & 0.904 & 0.450 & 0.406 & 0.992 & 0.982 & 0.722 &0.700 
			\\
			& & XIM-SIS3  & 72 & 271 & 0.950 & 0.878 & 0.436 & 0.378 & 0.992 & 0.966 & 0.704 &0.678 
			\\ 
			\addlinespace
			& \multirow[t]{12}{*}{300}& FAST  & 898 &1010 & 1.000 & 1.000 & 0.054 & 0.054 & 1.000 & 1.000 & 0.200 & 0.200 
			\\
			& & CRIS & 1336 &857 & 0.820 & 0.746 & 0.014 &0.012 & 0.894 & 0.858 & 0.082 & 0.074 
			\\
			& & CCRIS & 382 &616 & 0.974 & 0.972 & 0.102 & 0.102 & 0.992 & 0.990 & 0.426 & 0.426 
			\\
			& & CSIRS & 1116 &1116 & 1.000 & 1.000 & 0.040 & 0.040 & 1.000 & 1.000 & 0.156 & 0.156 
			\\
			& & WSIRS & 1094 &952 & 0.892 & 0.888 & 0.024 & 0.024 & 0.956 & 0.958 & 0.142 & 0.136 
			\\
			& & rCFRS & 1013 &1017 & 1.000 & 1.000 & 0.048 & 0.048 & 1.000 & 1.000 & 0.166 & 0.166 
			\\
			& & rCDCS & 22 &42 & 1.000 & 1.000 & 0.754 & 0.754 & 1.000 & 1.000 & 0.994 & 0.994 
			\\
			& & CRSIS & 1039 &1012 & 1.000 & 1.000 & 0.046 & 0.046 & 1.000 & 1.000 & 0.146 & 0.146 
			\\
			& & CDCS & 4 &1 & 1.000 & 1.000 & 0.998 & 0.998 & 1.000 & 1.000 & 1.000 & 1.000 
			\\
			& &  HSIC-SIS & 3 &1 & 1.000 & 1.000 & 1.000 & 1.000 & 1.000 & 1.000 & 1.000 & 1.000 
			\\
			& & XIM-SIS1  & 8 &37 &1.000 &1.000 &0.766 &0.766 &1.000 &1.000 &0.958 &0.958 
			\\
			& & XIM-SIS2  & 8 & 43 & 1.000 & 1.000 & 0.760 & 0.760 & 1.000 & 1.000 & 0.954 &0.954 
			\\
			& & XIM-SIS3  & 9 & 39 & 1.000 & 1.000 & 0.762 & 0.762 & 1.000 & 1.000 & 0.948 &0.948 \\
			\bottomrule    
\end{tabular}}
\end{table}

\begin{table}[!htbp]
	\setlength{\tabcolsep}{26pt} 
	\centering
	\caption{The quantiles of the minimum model size $S$ for Example \ref{ex:aft}}
	\label{tab: aft size}
	\resizebox{\textwidth}{!}{ 
		\begin{tabular}{llrrrrrrrrrr} 
  \toprule 
			\multicolumn{1}{l}{\multirow{2}{*}{CR}}    & \multicolumn{1}{l}{\multirow{2}{*}{Method}} & \multicolumn{5}{l}{$n=200$} & \multicolumn{5}{l}{$n=300$} \\ \cmidrule(lr){3-7} \cmidrule(lr){8-12}
   &   & 5\% & 25\% & 50\% & 75\% & \multicolumn{1}{l}{95\%}   & 5\% & 25\% & 50\% & 75\% & 95\%  \\ \hline \addlinespace 
			\multirow[t]{12}{*}{30\%}& FAST& 36 & 267 & 754 & 1394 & 1880 & 20 & 249 & 750 & 1260 &  1895 
\\
			& CRIS& 81 &  532 & 1067 & 1519 & 1947 & 85 & 557 & 1056 &1531 & 1895 
\\
			& CCRIS& 7 & 46 & 147 & 444 & 1130 & 5 & 18 & 75 & 253 &830 
\\
			& CSIRS& 60 & 410 & 1011 & 1665 & 1973 & 36 & 424 & 970 & 1534 &1939 
\\
			& WSIRS& 87 & 348 & 879 & 1482 & 1925 & 51 & 333 & 824 & 1340 &1881 
\\
			& rCFRS& 62 & 396 & 871 & 1403 & 1896 & 43 & 367 & 889 & 1438 &1844 
\\
			& rCDCS& 4 & 8 & 19 & 45 & 127 & 4 & 5 & 6 & 11 &29 
\\
			& CRSIS& 72 & 401 & 865 & 1482 & 1872 & 56 & 401 & 928 & 1468 &1873 
\\
			& CDCS& 3 & 3 & 4 & 6 & 18 & 3 & 3 & 4 & 4 &5 
\\
			& HSIC-SIS& 3 & 3 & 3 & 3 & 6 & 3 & 3 & 3 & 3 &4 
\\
			& XIM-SIS1   &3 &3 &4 &17 &185 &3 &3 &3 &3 &9 
\\
 & XIM-SIS2   & 3 & 3 & 5 & 19 & 163 & 3 & 3 & 3 & 3 &10 
\\
 & XIM-SIS3   & 3 & 3 & 5 & 21 & 170 & 3 & 3 & 3 & 3 &10 \\
 \addlinespace 
		\multirow[t]{12}{*}{50\%}	& FAST& 90 & 447 & 911 & 1490 & 1901 & 48 & 392 & 898 & 1402 &  1837 
\\
			& CRIS& 216 & 879 & 1349 & 1737 & 1974 & 222 & 823 & 1336 & 1680 & 1952 
\\
			& CCRIS& 27 & 191 & 507 & 889 & 1496 & 19 & 156 & 382 & 772 &1306 
\\
			& CSIRS& 119 & 574 & 1238 & 1738 & 1968 & 64 & 533 & 1116 & 1649 &1962 
\\
			& WSIRS& 128 & 639 & 1197 & 1609 & 1955 & 100 & 580 & 1094 & 1532 &1918 
\\
			& rCFRS& 71 & 419 & 941 & 1466 & 1913 & 62 & 478 & 1013 & 1495 &1909 
\\
			& rCDCS& 9 & 26 & 66 & 153 & 470 & 5 & 10 & 22 & 52 &137 
\\
			& CRSIS& 90 & 443 & 949 & 1473 & 1854 & 54 & 505 & 1039 & 1517 &1883 
\\
			& CDCS& 4 & 5 & 7 & 16 & 55 & 3 & 4 & 4 & 5 &10 
\\
			& HSIC-SIS& 3 & 3 & 3 & 4 & 11 & 3 & 3 & 3 & 4 &4 
\\
			& XIM-SIS1   &3 &12 &57 &256 &949 &3 &3 &8 &40 &242 
\\
 & XIM-SIS2   & 3 & 15 & 64 & 259 & 962 & 3 & 3 & 8 & 46 &279 
\\
 & XIM-SIS3   & 3 & 19 & 72 & 290 & 1109 & 3 & 4 & 9 & 43 &309 \\ 
			\bottomrule  
	\end{tabular}}
	\parbox{\linewidth}{\footnotesize 
    }
\end{table}

The results of the simulations are displayed in Table \ref{tab:aft} and  \ref{tab: aft size}. 
In this example, the linear effect signals significantly outweigh the quadratic effect signals, making the quadratic effect variable $X_7$ particularly challenging to detect. As evidenced by the results, the selection probability $\mathcal{P}_7$ remains substantially low.
 Consequently, the methods including FAST, CRIS, CCRIS, CSIRS, WSIRS, rCFRS, rCDCS, and CRSIS fail entirely across all parameter settings.
In contrast, CDCS, HISC-SIS, and our proposed XIM-SIS method maintain consistently high accuracy throughout the experiments. Notably, the selection probability  $\mathcal{P}_7$ rapidly converges to 1 as the sample size grows from 200 to 300. Furthermore, we observe that higher censoring rates consistently degrade the performance across all methods.

\begin{example}\label{ex:log exp}
	Consider the survival time $T$ generated from the nonlinear survival model:
	$$
	\log (T)=1.5-\exp \left(-X_1-0.8 X_2-X_7\right) * \varepsilon,
	$$
	where $\varepsilon \sim N(0, 1)$, $X \sim N_p(0, \boldsymbol{\Sigma})$ with covariance matrix $\Sigma=\left(\sigma_{l m}\right)_{p \times p}$, where $\sigma_{l l}=1$ and $\sigma_{l m}=0.5$ for $l \neq m$. Same as Example \ref{ex:cox}, set $n=200$ and $300$, $p=2000$, $CR=30\%, 50\%$.
\end{example}

\begin{table}[!htbp]
	\setlength{\tabcolsep}{21pt} 
	\centering
	\caption{The median and IQR of $\mathcal{S}$, the proportion of $\mathcal{P}_j$ and $\mathcal{P}_a$ for Example \ref{ex:log exp}}
	\label{tab:log exp}
	\resizebox{\textwidth}{!}{ 
		\begin{tabular}{lllrrllllllll}   \toprule 
			\multicolumn{1}{l}{\multirow{2}{*}{CR}}   & \multicolumn{1}{l}{\multirow{2}{*}{$n$}} & \multicolumn{1}{l}{\multirow{2}{*}{Method}} &  \multicolumn{1}{r}{\multirow{2}{*}{Median}}&\multicolumn{1}{r}{\multirow{2}{*}{IQR}}& \multicolumn{4}{l}{$d_1=n/log(n)$} & \multicolumn{4}{l}{$d_2=n-1$} \\ \cmidrule(lr){6-9} \cmidrule(lr){10-13}
			& &  & &&$\mathcal{P}_1$ &$\mathcal{P}_2$ &$\mathcal{P}_7$ &\multicolumn{1}{l}{$\mathcal{P}_a$}   &$\mathcal{P}_1$ &$\mathcal{P}_2$ &$\mathcal{P}_7$ & $\mathcal{P}_a$ \\ \hline \addlinespace
			\multirow[t]{24}{*}{20\%}& \multirow[t]{12}{*}{200}&  FAST & 805 &670 &0.180 &0.124 &0.180 &0.002 &0.440 &0.396 &0.450 &0.070 
			\\
			& & CRIS  & 1547 &614 &0.020 &0.028 & 0.022 & 0.000 & 0.118 & 0.110 & 0.100 &0.000 
			\\
			& & CCRIS  & 1988 &71 &0.002 & 0.000 & 0.006 & 0.000 & 0.012 & 0.012 & 0.012 & 0.000 
			\\
			& & CSIRS  & 7 &17 & 0.984 & 0.892 & 0.980 & 0.866 & 0.998 & 0.986 & 1.000 & 0.984 
			\\
			& & WSIRS & 7 &19 & 0.978 & 0.870 & 0.972 & 0.844 & 0.998 & 0.978 & 0.996 & 0.976 
			\\
			& & rCFRS & 1508 &634 &0.016 &0.020 &0.034 & 0.000 &0.150 &0.116 &0.144 & 0.000 
			\\
			& & rCDCS & 11 &24 & 0.974 & 0.844 &0.980 & 0.812 & 0.998 & 0.972 & 0.998 & 0.968 
			\\
			& & CRSIS & 1537 &664 & 0.020 & 0.028 & 0.030 & 0.000 & 0.142 & 0.116 & 0.136 & 0.002 
			\\
			& & CDCS & 176 &298 & 0.596 & 0.436 & 0.580 & 0.140 & 0.856 & 0.744 & 0.850 & 0.538 
			\\
			& &  HSIC-SIS & 16 &46 & 0.942 & 0.778 & 0.938 & 0.692 & 0.990 & 0.958 & 0.992 & 0.940 
			\\
			& & XIM-SIS1  & 4 &6 &0.992 &0.934 &0.992 &0.918 &1.000 &0.988 &1.000 &0.988 
			\\
			& & XIM-SIS2  & 4 & 7 & 0.992 & 0.928 & 0.994 & 0.916 & 1.000 & 0.986 & 1.000 &0.986 
			\\
			& & XIM-SIS3  & 4 & 6 & 0.996 & 0.930 & 0.994 & 0.920 & 1.000 & 0.986 & 1.000 &0.986 
			\\ 
			\addlinespace
			& \multirow[t]{12}{*}{300}& FAST  & 625 &745 & 0.292 & 0.244 & 0.314 & 0.030 & 0.656 & 0.576 & 0.646 & 0.256 
			\\
			& & CRIS & 1529 &633 & 0.022 & 0.038 & 0.036 & 0.000 & 0.162 & 0.152 & 0.170 & 0.004 
			\\
			& & CCRIS & 1996 &34 & 0.002 & 0.004 & 0.000 & 0.000 & 0.014 & 0.024 & 0.010 & 0.000 
			\\
			& & CSIRS & 3 &2 & 1.000 & 0.980 & 1.000 & 0.980 & 1.000 & 1.000 & 1.000 & 1.000 
			\\
			& & WSIRS & 3 &2 & 0.996 & 0.966 &0.994 & 0.964 &1.000 & 1.000 & 1.000 & 1.000 
			\\
			& & rCFRS & 1491 &661 & 0.044 & 0.034 & 0.042 & 0.000 & 0.200 & 0.166 & 0.230 & 0.020 
			\\
			& & rCDCS & 3 &3 & 0.998 & 0.976 & 1.000 & 0.974 & 1.000 & 1.000 & 1.000 & 1.000 
			\\
			& & CRSIS & 1499 &676 & 0.034 & 0.024 & 0.064 & 0.000 & 0.208 & 0.172 & 0.216 & 0.008 
			\\
			& & CDCS & 73 &163 & 0.836 & 0.638 & 0.830 & 0.426 & 0.970 & 0.888 & 0.968 & 0.838 
			\\
			& &  HSIC-SIS & 4 &10 & 0.990 & 0.938 & 0.990 & 0.918 & 1.000 & 0.994 & 1.000 & 0.994 
			\\
			& & XIM-SIS1  & 3 &0 &1.000 &1.000 &1.000 &1.000 &1.000 &1.000 &1.000 &1.000 
			\\
			& & XIM-SIS2  & 3 & 0 & 1.000 & 1.000 & 1.000 & 1.000 & 1.000 & 1.000 & 1.000 &1.000 
			\\
			& & XIM-SIS3  & 3 & 0 & 1.000 & 0.998 & 1.000 & 0.998 & 1.000 & 1.000 & 1.000 &1.000 
			\\ 
			\hline
			\addlinespace 
			\multirow[t]{24}{*}{40\%}& \multirow[t]{12}{*}{200}& FAST  & 1453 &742 & 0.034 & 0.040 & 0.044 & 0.000 & 0.166 & 0.172 & 0.162 & 0.004 
			\\
			& & CRIS & 1563 &597 & 0.026 &0.028 & 0.022 &0.000 &0.104 & 0.106 &0.100 &0.000 
			\\
			& & CCRIS & 1235 &754 & 0.096 & 0.088 & 0.094 & 0.000 & 0.296 & 0.242 & 0.272 & 0.016 
			\\
			& & CSIRS & 58 &134 & 0.812 & 0.600 & 0.794 & 0.398 & 0.964 & 0.882 & 0.954 & 0.810 
			\\
			& & WSIRS & 147 &278 & 0.604 & 0.462 &0.590 & 0.188 & 0.852 & 0.746 & 0.854 & 0.576 
			\\
			& & rCFRS & 1318 &725 & 0.068 &0.042 & 0.064 &0.000 & 0.230 & 0.198 & 0.214 & 0.002 
			\\
			& & rCDCS & 131 &235 & 0.616 & 0.354 &0.598 & 0.170 & 0.888 & 0.778 & 0.920 &0.642 
			\\
			& & CRSIS & 1337 &744 & 0.062 & 0.040 &0.056 & 0.000 & 0.246 & 0.186 & 0.218 & 0.002 
			\\
			& & CDCS & 49 &133 & 0.808 & 0.610 & 0.822 & 0.428 & 0.964 & 0.870 & 0.968 & 0.814 
			\\
			& &  HSIC-SIS & 126 &237 & 0.662 & 0.404 & 0.632 & 0.166 & 0.910 & 0.772 & 0.898 & 0.634 
			\\
			& & XIM-SIS1  & 41 &114 &0.848 &0.662 &0.868 &0.494 &0.952 &0.886 &0.980 &0.830 
			\\
			& & XIM-SIS2  & 41 & 123 & 0.848 & 0.650 & 0.862 & 0.474 & 0.954 & 0.880 & 0.978 &0.822 
			\\
			& & XIM-SIS3  & 48 & 128 & 0.840 & 0.634 & 0.860 & 0.462 & 0.962 & 0.874 & 0.974 &0.818 
			\\ 
			\addlinespace
			& \multirow[t]{12}{*}{300}& FAST  & 1355 &735 & 0.060 & 0.064 & 0.082 & 0.000 & 0.266 & 0.266 & 0.260 & 0.026 
			\\
			& & CRIS & 1622 &568 & 0.026 & 0.032 & 0.030 &0.000 & 0.166 & 0.164 & 0.146 & 0.004 
			\\
			& & CCRIS & 1127 &772 & 0.138 & 0.108 & 0.160 & 0.004 & 0.430 & 0.346 & 0.446 & 0.068 
			\\
			& & CSIRS & 14 &40 & 0.974 & 0.820 & 0.966 & 0.772 & 0.998 & 0.980 & 1.000 & 0.978 
			\\
			& & WSIRS & 54 &167 & 0.836 & 0.614 & 0.832 & 0.498 & 0.958 & 0.884 & 0.970 & 0.848 
			\\
			& & rCFRS & 1201 &792 & 0.124 & 0.074 & 0.118 & 0.004 & 0.342 & 0.284 & 0.368 & 0.042 
			\\
			& & rCDCS & 44 &101 & 0.888 & 0.674 & 0.892 & 0.548 & 0.992 & 0.952 & 0.990 & 0.936 
			\\
			& & CRSIS & 1233 &804 & 0.110 & 0.070 & 0.106 & 0.002 & 0.346 & 0.278 & 0.370 & 0.028 
			\\
			& & CDCS & 13 &35 & 0.980 & 0.822 & 0.982 & 0.798 & 0.998 & 0.986 & 1.000 & 0.984 
			\\
			& &  HSIC-SIS & 43 &92 & 0.910 & 0.666 & 0.906 & 0.566 & 0.984 & 0.950 & 0.992 & 0.926 
			\\
			& & XIM-SIS1  & 7 &23 &0.988 &0.906 &0.982 &0.882 &0.998 &0.986 &0.998 &0.982 
			\\
			& & XIM-SIS2  & 8 & 24 & 0.988 & 0.888 & 0.982 & 0.862 & 0.996 & 0.986 & 0.998 &0.982 
			\\
			& & XIM-SIS3  & 8 & 25 & 0.986 & 0.892 & 0.978 & 0.862 & 0.998 & 0.988 & 0.998 &0.984 \\ 
			\bottomrule   
\end{tabular}}
\end{table}

\begin{table}[!htbp]
	\setlength{\tabcolsep}{26pt} 
	\centering
	\caption{The quantiles of the minimum model size $S$ for Example \ref{ex:log exp}}
	\label{tab: log exp size}
	\resizebox{\textwidth}{!}{ 
		\begin{tabular}{llrrrrrrrrrr} 
  \toprule 
			\multicolumn{1}{l}{\multirow{2}{*}{CR}}    & \multicolumn{1}{l}{\multirow{2}{*}{Method}} & \multicolumn{5}{l}{$n=200$} & \multicolumn{5}{l}{$n=300$} \\ \cmidrule(lr){3-7} \cmidrule(lr){8-12}
   &   & 5\% & 25\% & 50\% & 75\% & \multicolumn{1}{l}{95\%}   & 5\% & 25\% & 50\% & 75\% & 95\%  \\ \hline \addlinespace 
			\multirow[t]{12}{*}{20\%}& FAST& 175 & 466 & 805 & 1136 & 1711 & 83 & 290 & 625 & 1035 &  1556 
\\
			& CRIS& 729 &  1211 & 1547 & 1824 & 1962 & 662 & 1159 & 1529 &1792 & 1963 
\\
			& CCRIS& 1659 & 1928 & 1988 & 1999 & 2000 & 1800 & 1966 & 1996 & 2000 &2000 
\\
			& CSIRS& 3 & 4 & 7 & 21 & 107 & 3 & 3 & 3 & 5 &25 
\\
			& WSIRS& 3 & 4 & 7 & 23 & 118 & 3 & 3 & 3 & 5 &37 
\\
			& rCFRS& 608 & 1162 & 1508 & 1796 & 1961 & 583 & 1106 & 1491 & 1766 &1961 
\\
			& rCDCS& 3 & 4 & 11 & 28 & 128 & 3 & 3 & 3 & 6 &30 
\\
			& CRSIS& 593 & 1142 & 1537 & 1806 & 1975 & 571 & 1099 & 1499 & 1775 &1961 
\\
			& CDCS& 18 & 70 & 176 & 368 & 896 & 6 & 24 & 73 & 187 &627 
\\
			& HSIC-SIS& 3 & 5 & 16 & 51 & 229 & 3 & 3 & 4 & 13 &69 
\\
			& XIM-SIS1   &3 &3 &4 &9 &54 &3 &3 &3 &3 &8 
\\
 & XIM-SIS2   & 3 & 3 & 4 & 10 & 54 & 3 & 3 & 3 & 3 &8 
\\
 & XIM-SIS3   & 3 & 3 & 4 & 9 & 61 & 3 & 3 & 3 & 3 &8 \\  
\addlinespace 
		\multirow[t]{12}{*}{40\%}	& FAST& 476 & 1029 & 1453 & 1771 & 1953 & 433 & 945 & 1355 & 1680 &  1943 
\\
			& CRIS& 803 & 1204 & 1563 & 1800 & 1971 & 693 & 1271 & 1622 & 1839 & 1966 
\\
			& CCRIS& 328 & 850 & 1235 & 1604 & 1929 & 261 & 754 & 1127 & 1525 &1919 
\\
			& CSIRS& 5 & 18 & 58 & 152 & 474 & 3 & 5 & 14 & 45 &209 
\\
			& WSIRS& 12 & 56 & 147 & 334 & 884 & 5 & 17 & 54 & 184 &591 
\\
			& rCFRS& 448 & 919 & 1318 & 1644 & 1957 & 346 & 797 & 1201 & 1589 &1903 
\\
			& rCDCS& 15 & 57 & 131 & 292 & 687 & 5 & 15 & 44 & 116 &336 
\\
			& CRSIS& 502 & 893 & 1337 & 1637 & 1924 & 368 & 790 & 1233 & 1594 &1892 
\\
			& CDCS& 5 & 18 & 49 & 151 & 470 & 3 & 5 & 13 & 40 &183 
\\
			& HSIC-SIS& 13 & 56 & 126 & 293 & 814 & 4 & 14 & 43 & 106 &460 
\\
			& XIM-SIS1   &4 &13 &41 &127 &524 &3 &4 &7 &27 &136 
\\
 & XIM-SIS2   & 4 & 13 & 41 & 136 & 520 & 3 & 4 & 8 & 28 &129 
\\
 & XIM-SIS3   & 4 & 13 & 48 & 141 & 509 & 3 & 4 & 8 & 29 &131 \\ 
			\bottomrule  
	\end{tabular}}
	\parbox{\linewidth}{\footnotesize 
    }
\end{table}

The simulation results are presented in Table \ref{tab:log exp} and  \ref{tab: log exp size}. From these results, we can find that FAST, CRIS, CCRIS, rCFRS and  CRSIS completely fail whether  $n=200$ or $300$, and $CR=20\%$, or $40\%$.  Moreover, their performance could not be enhanced by either increasing the sample size or decreasing the censoring rate.   For our proposed method "XIM-SIS", it surpasses other models across all parameter configurations, also boasting the smallest median and the narrowest interquartile range (IQR) of the minimum model size $\mathcal{S}$.

\section{Real Data Analysis}

This dataset includes expression values of 24481 genes in 295 consecutive women with breast cancer. Diagnoses were made between 1984 and 1995, and the data were sourced from the fresh-frozen tissue bank of the Netherlands Cancer Institute \citep{van2002gene}. We utilized a filtered version by \cite{michiels2005prediction} reducing to 4948 genes. The dataset is downloaded from the R package \textit{‘cancerdata’}. At the end of the follow-up, 216 of the 295 patients were still alive, yielding a censoring rate of $73.22 \%$. The duration until death or censoring spanned from 0.05 to 18.3 years, with a median duration of 7.2 years. In this dataset, there are 6 samples with 21 missing gene expressions. We use the same method as \cite{chen2018robust}, impute the missing values using the weighted $K$-nearest neighbor method \citep{troyanskaya2001missing} with $K=15$. 

We randomly divide  the data into a training set with a sample size of  $n_{\text {tra }}=150$ and a test set with a sample size of $n_{\text {tes }}=142$,  repeating this procedure 100 times.  For each split,  the selected model size is considered as $\lceil n_{\text{tra}} / \log n_{\text{tra}}\rceil$ in the training model. After  the screening procedure, we built a Cox Proportional Hazard model with the penalized partial likelihood and LASSO to further eliminate irrelevant variables. LASSO is used due to considering that some truly unimportant predictors may still be retained in the screening stage. We use 10-fold cross validation to select the tuning parameter in the penalization procedure.

We apply the proposed XIM-SIS method alongside the other ten approaches considered in Section 3 to this real data. To evaluate the predictive performance  of  different  methods in the censored survival data,  we use  the C-statistic \citep{uno2011c},  which  is defined as:
$$C_n=P\left(\hat{\boldsymbol{\beta}}^{\top} \boldsymbol{X}_i>\hat{\boldsymbol{\beta}}^{\top} \boldsymbol{X}_j \mid T_i<T_j\right)$$
where $\boldsymbol{X}_i$ and $\boldsymbol{X}_j$ represent the selected variables, $T_i$ and $T_j$ denote the survival time for the $i$th and $j$th patients, subject to right censoring, and $\hat{\boldsymbol{\beta}}$ is the parameter estimate in the resulting Cox model. The term $\hat{\boldsymbol{\beta}}^{\top} \boldsymbol{X}_i$ is the risk score for the $i$th patient. Specifically, interpretations of C- statistics are as follows: $C_n=0.5$ implies no predictive power of the model; $0.5<C_n \leq 0.7$ implies moderate predictive power; $0.7<C_n \leq 0.9$ implies moderate to strong predictive power; and $0.9<C_n \leq 1$ implies strong predictive power \citep{zhou2017model}. We repeat the procedures 100 times and report the average C-statistics and standard errors in Tables \ref{tab: Breast}-\ref{tab: Breast gene lasso}.

\begin{table}[htbp]
		\centering
		\caption{Average and SD of C-Statistic for Different Screening Procedures on Breast Cancer Data}
		\label{tab: Breast}
		\resizebox{\textwidth}{!}{ 
		\begin{tabular}{llllllllllllll}
			\toprule Method & FAST & CRIS & CCRIS & CSIRS & WSIRS & rCFRS & rCDCS & CRSIS & CDCS &  HSIC-SIS & XIM-SIS1 & XIM-SIS2 & XIM-SIS3 \\ \hline \addlinespace
			C-statistic & 0.790 & 0.796 & 0.780 & 0.786 & 0.794 & 0.785 & 0.786 & 0.797 & 0.778 & 0.796 
&   0.792 &0.799 &0.790 
\\
			SD & 0.047 & 0.044 & 0.050 & 0.047 & 0.043 & 0.049 & 0.049 & 0.039 & 0.045 & 0.043 &   0.046 &0.045 &0.049 \\ \bottomrule
		\end{tabular}}
  \parbox{\linewidth}{\footnotesize We choose $M=[\sqrt{n_{tra}}]-1=11$ for XIM-SIS1, $M=[\sqrt{n_{tra}}]=12$ for XIM-SIS2, $M=[\sqrt{n_{tra}}]+1=13$ for XIM-SIS3}
	\end{table}

\begin{table}[htbp]
		\centering
		\caption{Average and SD of C-Statistic for Different Screening Procedures with LASSO on Breast Cancer Data}
		\label{tab: Breast lasso}
		\resizebox{\textwidth}{!}{ 
		\begin{tabular}{llllllllllllll}
			\toprule Method & FAST & CRIS & CCRIS & CSIRS & WSIRS & rCFRS & rCDCS & CRSIS & CDCS &HSIC-SIS &   XIM-SIS1&XIM-SIS2&XIM-SIS3 \\ \hline 
            \addlinespace
			C-statistic & 0.627 & 0.631 & 0.522 & 0.677 & 0.576 & 0.639 & 0.643 & 0.660 & 0.671 & 0.646 &   0.638 &	0.650 &	0.644 
\\
			SD & 0.091 & 0.082 & 0.050 & 0.075 & 0.081 & 0.089 & 0.089 & 0.084 & 0.080 & 0.084 & 0.091 &	0.086 &	0.090 \\ \bottomrule
		\end{tabular}}
  \parbox{\linewidth}{\footnotesize We choose $M=[\sqrt{n_{tra}}]-1=11$ for XIM-SIS1, $M=[\sqrt{n_{tra}}]=12$ for XIM-SIS2, $M=[\sqrt{n_{tra}}]+1=13$ for XIM-SIS3}
 \end{table}


\begin{table}[htbp]
		\centering
		\caption{The names of the top 20 most frequently selected genes by various methods followed by LASSO for breast cancer data}
		\label{tab: Breast gene lasso}
		\resizebox{\textwidth}{!}{ 
		\begin{tabular}{lllllll}
			\toprule  FAST & CRIS & CCRIS & CSIRS & WSIRS & rCFRS  &rCDCS \\ \hline 
   \addlinespace
			NM\_001168& NM\_001333& Contig48328\_RC& NM\_001333& NM\_001333& Contig38288\_RC &Contig38288\_RC\\
			NM\_001333& Contig38288\_RC& NM\_004418& Contig48270\_RC& U96131& NM\_001333 &NM\_001605\\
			Contig55725\_RC& Contig48270\_RC& NM\_006472& Contig38288\_RC& NM\_006399& NM\_001605 &Contig31288\_RC\\
			Contig56390\_RC& Contig31288\_RC& AF279865& NM\_001605& Contig36879\_RC& Contig31288\_RC &D14678\\
			AL049265& D43950& Contig43806\_RC& NM\_002811& D43950& D14678 &NM\_001333\\
			NM\_003258& D14678& Contig749\_RC& D43950& U82987& D43950 &D43950\\
			NM\_018410& NM\_001605& AB020689& NM\_016359& NM\_006623& NM\_006607 &NM\_003258\\
			NM\_001085& NM\_016359& AB037844& U96131& NM\_012067& NM\_000270 &NM\_003981\\
			NM\_001124& U96131& AB037863& NM\_003295& Contig55188\_RC& Contig48270\_RC &NM\_005733\\
			Contig48328\_RC& Contig36879\_RC& AF052169& D14678& NM\_004418& D38553 &NM\_006607\\
			NM\_004701& D38553& AF167706& Contig31288\_RC& NM\_006579& NM\_003430 &NM\_016359\\
			NM\_005375& M96577& AL110171& AF148505& AB023166& NM\_016359 &D38553\\
			NM\_000125& AL137347& Contig30263\_RC& NM\_003258& Contig31288\_RC& NM\_005733 &Contig48270\_RC\\
			NM\_001809& NM\_000291& Contig36106\_RC& NM\_006461& NM\_001168& NM\_007057 &NM\_000270\\
			NM\_020974& NM\_003258& Contig39556\_RC& NM\_006579& NM\_004358& NM\_016109 &NM\_007057\\
			AL137566& NM\_006096& Contig41530\_RC& Contig36879\_RC& NM\_017522& NM\_003258 &Contig48328\_RC\\
			NM\_001394& Contig57584\_RC& Contig42011\_RC& M96577& Contig55725\_RC& NM\_005689 &NM\_001109\\
			Contig41413\_RC& NM\_001124& Contig43871\_RC& NM\_001124& Contig55771\_RC& NM\_007019 &NM\_005689\\
			Contig53307\_RC& NM\_003981& Contig44064\_RC& NM\_004553& NM\_002426& NM\_018410 &NM\_006096\\
			NM\_005733& NM\_006461& Contig44553\_RC& NM\_005804& NM\_003600& Contig48328\_RC &NM\_007019\\
			\bottomrule 
   \addlinespace\addlinespace
   \toprule  CRSIS & CDCS & HSIC & XIM-SIS1& XIM-SIS2&  XIM-SIS3&\\ \hline 
   \addlinespace 
			NM\_001333& NM\_001333& NM\_002811& Contig31288\_RC& Contig31288\_RC&  Contig38288\_RC
&\\
			NM\_001605& D43950& U96131& NM\_001605& Contig38288\_RC&  NM\_001333
&\\
			Contig58368\_RC& Contig38288\_RC& D43950& Contig38288\_RC& NM\_001605&  Contig31288\_RC
&\\
			Contig38288\_RC& NM\_001605& Contig38288\_RC& Contig58368\_RC& NM\_001333&  NM\_001605
&\\
			Contig31288\_RC& NM\_002811& NM\_003981& NM\_001333& Contig58368\_RC&  Contig58368\_RC
&\\
			D14678& Contig48270\_RC& D14678& D14678& D14678&  D14678
&\\
			D43950& U96131& M96577& NM\_000270& NM\_000270&  NM\_016359
&\\
			NM\_000270& Contig31288\_RC& AF148505& D43950& Contig17359\_RC&  NM\_002689
&\\
			NM\_005689& NM\_016359& NM\_016359& NM\_003981& D43950&  NM\_006096
&\\
			NM\_006607& D14678& Contig31288\_RC& D38553& NM\_002689&  D38553
&\\
			Contig48270\_RC& NM\_003295& Contig48270\_RC& NM\_016359& Contig48270\_RC&  D43950
&\\
			NM\_005733& NM\_003258& NM\_001333& Contig48270\_RC& D38553&  NM\_000270
&\\
			NM\_003981& M96577& NM\_003295& NM\_002689& NM\_003981&  NM\_003981
&\\
			Contig1982\_RC& AF148505& NM\_001168& NM\_006096& NM\_006096&  NM\_005733
&\\
			Contig27386\_RC& Contig55725\_RC& NM\_006607& Contig39891\_RC& NM\_016359&  Contig17359\_RC
&\\
			NM\_001168& NM\_005804& Contig55725\_RC& NM\_000291& NM\_000291&  Contig48270\_RC
&\\
			NM\_002689& NM\_012291& AL049265& NM\_005733& NM\_001168&  NM\_000291
&\\
			NM\_003430& Contig36879\_RC& NM\_012291& NM\_005804& NM\_005733&  Contig57584\_RC
&\\
			NM\_003662& Contig48328\_RC& Contig56390\_RC& U96131& NM\_006558&  NM\_001168
&\\
			NM\_005375& Contig56390\_RC& NM\_001124& Contig17359\_RC& NM\_007057&  NM\_004217&\\
			\bottomrule
		\end{tabular}}
		\parbox{\linewidth}{\footnotesize The genes are ranked by the frequency of screening result.}
	\end{table}

Table \ref{tab: Breast} is the C-statistics for various screening methods under the working Cox model for the breast data after screening. As C-statistics are all large than 0.7, indicating all these methods has strong predictive power. In addition, our proposed method XIM-SIS has the best performance with a C-statistics value 0.799. Table \ref{tab: Breast lasso} is the C-statistics for various screening methods with LASSO, indicating moderate predictive power with values ranging from 0.533 to 0.677. Our proposed method achieves a C-statistic of 0.650, slightly below CSIRS, CRSIS, and CDCS. Table \ref{tab: Breast gene lasso} summarized the top 20 total frequency screened genes with LASSO penalty, totaling 92 unique genes discovered across different methods. Among these,
NM\_001333, Contig31288\_RC, D43950 Contig38288\_RC, Contig48270\_RC, and D14678 stand out as the top six most frequently selected genes, appearing in 12, 11, 11, 10, 10, and 10 approaches, respectively. It should be noted that these six genes are all successfully selected by our method XIM-SIS with top frequency no matter the value of M.  This demonstrates the efficacy of our method in the application field. In addition, although the ranking of features may vary slightly with different selections of M, such slight variations do not impede our method's ability to identify the most significant features. Upon literature review, we found that Contig38288\_RC was previously identified as a predictive gene in \citep{van2002gene}, which is the top three selected gene by our proposed method XIM-SIS for all three M option. This underscores the effectiveness of our screening approach in providing meaningful results.

\section{Conclusion and Discussion}
\label{sec:conc}
This work introduces a straightforward model-free technique for feature screening in ultra-high dimensional censored data, which may find widespread use in contemporary disciplines such as medical research. Simulation studies provide sufficient evidence supporting the effectiveness of our method across diverse practical scenarios, especially in the case with small censoring rate and enough sample numbers. We additionally demonstrate the efficacy of our proposal through an analysis of gene association studies focusing on the survival time of patients with breast cancer. In this paper, our proposed screening method focuses solely on identifying main effects, while detecting interaction effects among predictors and their impact on the response variable is also a crucial step for ultra-high dimensional censored data, which leave space here for future work.

\section*{Acknowledgements}
This study was partially supported by the Guangdong Higher Education Upgrading Plan (2021-2025) of "Rushing to the Top, Making Up Shortcomings and Strengthening Special Features" with No. of 2023KTSCX159, and the Guangdong Provincial Key Laboratory of Interdisciplinary Research and Application for Data Science (project code 2022B1212010006).


\appendix

\section{Proof of Main Theorems}
The appendix contains additional  proofs of the main theorems.

\begin{lemma}[] \label{lemma.1}
	Consider a pair of continuous random variables $(X, Y)$, then for any $\epsilon >0$ and fixed $M$, there exist $N \in \mathds{N}$, such that for all $n>N$, there exists a constant $\Gamma >0$ such that
	\begin{align}
		P\Bigl(\Bigl|\xi_{n,M}(X, Y)-\xi(X, Y)\Bigr| \geqslant \epsilon \Bigr) \leqslant 2 \exp \left(-\Gamma \frac{\epsilon^2 n^2-4\epsilon   }{576 n}\right)  \notag
	\end{align}
\end{lemma}

In order to prove Lemma \ref{lemma.1}, we need the following inequality given in Lemma A.11 in \cite{chatterjee2021new}. Let $G(y) \equiv E\{\mathds{1}(Y \geq y)\}$ and $G_X(y) \equiv E\{\mathds{1}(Y \geq y) \mid X\}$. Let $Q(X, Y) \equiv \int \operatorname{var}\left(G_X(y)\right) \mathrm{d} F_Y(y)$ and $U(Y) \equiv \int G(y)(1-G(y)) \mathrm{d} F_Y(y)$. Then according to the definition of $\xi$, we can form $\xi=Q/U$. Following \cite{lin2023boosting}, further denote 
\begin{align}
	& Q_{n, M}(X, Y) \equiv \frac{1}{n M} \sum_{i=1}^n \sum_{m=1}^M \min \left\{F_n\left(Y_i\right) ,F_n\left(Y_{j_m(X_i)}\right)\right\}-\frac{1}{n} \sum_{i=1}^n G_n\left(Y_i\right)^2 \notag \\
	& U_n(Y) \equiv \frac{1}{n} \sum_{i=1}^n G_n\left(Y_i\right)\left(1-G_n\left(Y_i\right)\right)=\frac{1}{6}\left(1-\frac{1}{n^2}\right)\notag
\end{align}
where $F_n(y) \equiv \frac{1}{n} \sum_{i=1}^n \mathds{1}\left(Y_i \leq y\right)$ and $G_n(y) \equiv \frac{1}{n} \sum_{i=1}^n \mathds{1}\left(Y_i \geq y\right)$ are the empirical counterparts of $F$ and $G$, respectively. From the proof of the theorem 2.1 of \cite{lin2023boosting}, we also have 
\begin{align}
	& \frac{Q_{n, M}(X, Y)}{U_n(Y)} =\frac{n+(M+1) / 4}{n-1} \xi_{n, M}(X,Y)+\frac{M-1}{2(n-1)} \notag
\end{align}
that is
\begin{align}
	\xi_{n, M}(X,Y) = \frac{4(n-1)}{4n+(M+1)}\dfrac{Q_{n, M}(X,Y)}{U_n(Y)}-\frac{2(M-1)}{4n+(M+1)} \notag  
\end{align}

\begin{proof}[\textbf{Proof of Lemma \ref{lemma.1}:}]
	\begin{align}
		& P\left(\Bigl|\xi_{n,M}(X, Y)-\xi(X, Y)\Bigr| \geqslant \epsilon \right) \notag \\
		= & P\left(\left|\xi_{n,M}(X, Y)-\dfrac{Q_{n, M}(X,Y)}{U_n(Y)} + \dfrac{Q_{n, M}(X,Y)}{U_n(Y)} - \xi(X, Y)\right| \geqslant \epsilon \right) \notag \\
		\leqslant & P\left(\left|\xi_{n,M}(X, Y)-\dfrac{Q_{n, M}(X,Y)}{U_n(Y)}\right| \geqslant \frac{\epsilon}{2} \right) + P\left(\left|\dfrac{Q_{n, M}(X,Y)}{U_n(Y)}-\xi(X, Y)\right| \geqslant \frac{\epsilon}{2} \right) \notag \\
		\equiv & L_{1}+L_{2}
	\end{align}
	The remaining derivation is to examine $L_{1}$, and $L_{2}$ respectively.
	
	\textbf{Step L.1:} We first examine $L_{1}$. Since $1/n \sum_{i=1}^n G_n\left(Y_i\right)^2 = 1/n\sum_{i=1}^n \left(i/n\right)^2 $, then
	\begin{align}
		Q_{n, M}(X,Y) & =\frac{1}{n M} \sum_{i=1}^n \sum_{m=1}^M \min \left\{F_n\left(Y_i\right) F_n\left(Y_{j_m(X_i)}\right)\right\}-\frac{1}{n} \sum_{i=1}^n G_n\left(Y_i\right)^2 \notag \\
		& \leqslant \frac{1}{n M} \sum_{i=1}^n \sum_{m=1}^M F_n(Y_i)-\frac{(n+1)(2 n+1)}{6 n^2} = \frac{1+n}{2n}-\frac{(n+1)(2 n+1)}{6 n^2} = \frac{n^2-1}{6n^2} \notag \\
		& = U_n(Y) \notag
	\end{align}
	Then
	\begin{align}
		\left\lvert\,\left.\xi_{n,M}(X,Y)-\dfrac{Q_{n, M}(X,Y)}{U_n(Y)} \right\rvert\,\right. = & \left|\frac{4(n-1)}{4n+(M+1)}\dfrac{Q_{n, M}(X,Y)}{U_n(Y)}-\frac{2(M-1)}{4n+(M+1)}-\dfrac{Q_{n, M}(X,Y)}{U_n(Y)}\right| \notag \\
		= & \left|\frac{M+5}{4n+(M+1)} \cdot \dfrac{Q_{n, M}(X,Y)}{U_n(Y)}+\frac{2(M-1)}{4 n+(M+1)}\right| \notag \\
		\leqslant & \left|\frac{M+5}{4n+(M+1)}+\frac{2(M-1)}{4 n+(M+1)}\right| = \left|\frac{3(M+3)}{4n+(M+1)}\right| \leqslant \frac{M+1}{n} \notag
	\end{align}
	Since $\frac{M+1}{n} \rightarrow 0$ as $n \rightarrow \infty$, then there exist $N_1 \in \mathbb{N}$,  for any $n>N_1$,  we have
	\begin{align}
		P\left(\left|\xi_{n,M}(X, Y)-\dfrac{Q_{n, M}(X,Y)}{U_n(Y)}\right| \geqslant \frac{\epsilon}{2} \right) =0 \label{l.1}
	\end{align}
	
	\textbf{Step L.2:} Examine $L_{2}$. From the proof of Theorem 1.1 in \cite{chatterjee2021new},  $S_n \xrightarrow{\text { a.s. }} S$. And since $G(Y)=1-F(Y) \sim$ Uniform $[0,1]$. Thus,
	$$U(Y) \equiv \int G(y)(1-G(y)) \mathrm{d} F_Y(y) =\int F(y)(1-F(y)) \mathrm{d} F_Y(y) = \int_0^1 x(1-x) d x=\frac{1}{6}$$
	Then
	\begin{align}
		\left|\dfrac{Q_{n, M}(X,Y)}{U_n(Y)}-\frac{Q}{U}\right|= & \left|\frac{Q_{n, M}(X,Y)\left(U-U_n(Y)\right)+ U_n(Y) \left(Q_{n, M}(X,Y)-Q\right) }{U_n(Y) U}\right| \notag \\
		\leqslant & \left|\frac{U-U_n(Y)}{U}\right| + \left|\frac{Q_{n, M}(X,Y)-Q}{U}\right| 
		=  \frac{1}{n^2} + 6\left|Q_{n, M}(X,Y)-Q\right| \notag 
	\end{align}
	Therefore, we have
	\begin{align}
		& P\left(\left|\dfrac{Q_{n, M}(X,Y)}{U_n(Y)}-\frac{Q}{U}\right| \geqslant \frac{\epsilon}{2} \right) \notag \\
		\leqslant & P\left( 6\left|Q_{n, M}(X,Y)-Q\right| \geqslant \frac{\epsilon}{2} - \frac{1}{n^2} \right)
		=  P\left( \left|Q_{n, M}(X,Y)-Q\right| \geqslant \frac{\epsilon n^2-2}{12 n^2} \right) \notag \\    
		\leqslant & P\left(\left|Q_{n, M}(X, Y)-E\left(Q_{n, M}(X, Y)\right)\right| \geqslant \frac{\epsilon n^2-2}{24 n^2} \right) + P\left(\left|E\left(Q_{n, M}(X, Y)\right) -Q\right| \geqslant \frac{\epsilon n^2-2}{24 n^2} \right) \notag \\
		\leqslant & 2 \exp \left(-\Gamma \frac{\epsilon^2 n^4-4\epsilon n^2 +4  }{576 n^3}\right) \notag \\  
		\leqslant & 2 \exp \left(-\Gamma \frac{\epsilon^2 n^2-4\epsilon   }{576 n}\right)
		\label{l.2}
	\end{align}
	The last inequality is given from the proof of Theorem 1 in \cite{lin2023boosting}, 
	$\lim _{n \rightarrow \infty} E\left(Q_{n, M}\right)=Q$
	and  the Lemma A.11 in \cite{chatterjee2021new}, essentially bounded difference inequality, there exists a constant $\Gamma >0$ such that for any $n$ and $t>0$,
	$$
	P\left(\left|Q_{n, M}-E\left(Q_{n, M}\right)\right| \geq t\right) \leqslant 2 \exp \left(-\Gamma n t^2\right) 
	$$
	Combine \ref{l.1} and \ref{l.2}, it is clear to see 
	\begin{align}
		P\Bigl(\Bigl|\xi_{n,M}(X, Y)-\xi(X, Y)\Bigr| \geqslant \epsilon \Bigr) \leqslant & 2 \exp \left(-\Gamma \frac{\epsilon^2 n^2-4\epsilon   }{576 n}\right) \notag 
	\end{align}
	Which end the proof of  lemma \ref{lemma.1}.
\end{proof}

\begin{lemma}[] \label{lemma.2}
	Let $\left\{T_i\right\}$ and $\left\{Z_i\right\}$ be no tie sequences of independently identically distributed nonnegative random variables with distribution functions $F(T)$ and $F(Z)$, respectively. Let $\hat{F}(Z)$ be the empirical cumulative distribution function of $\left\{Z_i\right\}$, $S(T) = 1 - F(T)$, and $\hat{S}(T)$ be the Kaplan-Meier estimator of $S(T)$. Then 
	\begin{align}
		\xi_{n,M}\left(\hat{F}(Z), \hat{S}(T)\right) = \xi_{n,M}\left(F(Z), S(T)\right), \quad \xi_{n,M}\left(\hat{S}(T), \hat{F}(Z)\right) = \xi_{n,M}\left(S(T), F(Z)\right).
	\end{align}
\end{lemma}

\begin{proof}[\textbf{Proof of Lemma \ref{lemma.2}:}]
	We first prove $\xi_{n,M}\left(\hat{F}(Z), \hat{S}(T)\right) = \xi_{n,M}\left(F(Z), \hat{S}(T)\right)$. Since $\left\{Z_i\right\}$ is a no tie sequences, then for any integer $i, j < n$, we have 
	\begin{align}
		\sum_{l=1}^n \mathds{1}\left(\hat{F}(Z_{i})<\hat{F}(Z_{l}) \leqslant \hat{F}(Z_{j})\right) = \sum_{l=1}^n \mathds{1}\left(Z_{i}<Z_{l} \leqslant Z_{j}\right) = \sum_{l=1}^n \mathds{1}\left(F(Z_{i})<F(Z_{l}) \leqslant F(Z_{j})\right) \notag
	\end{align}
	that is $j_m\left(\hat{F}(Z_i)\right)= j_m\left(F(Z_i)\right)$. Then
	\begin{align}
		Q_{n, M}\left(\hat{F}(Z), \hat{S}(T)\right) \equiv 
		& \frac{1}{n M} \sum_{i=1}^n \sum_{m=1}^M  \min \left\{F_n\left(\hat{S}\left(T_i\right)\right), F_n\left(\hat{S}\left(T_{j_m\left(\hat{F}\left(Z_i\right)\right)}\right)\right)\right\} -\frac{1}{n} \sum_{i=1}^n G_n\left(\hat{S}(T_i)\right)^2 \notag \\
		= & \frac{1}{n M} \sum_{i=1}^n \sum_{m=1}^M  \min \left\{F_n\left(\hat{S}\left(T_i\right)\right), F_n\left(\hat{S}\left(T_{j_m\left(F\left(Z_i\right)\right)}\right)\right)\right\} -\frac{1}{n} \sum_{i=1}^n G_n\left(\hat{S}(T_i)\right)^2 \notag \\
		= & Q_{n, M}\left(F(Z), \hat{S}(T)\right) \notag
	\end{align}
	According to the definition of $\xi_{n,M}$, we have
	\begin{align}
		\xi_{n,M}\left(\hat{F}(Z), \hat{S}(T)\right) =\dfrac{Q_{n, M}\left(\hat{F}(Z), \hat{S}(T)\right)}{U_n\left(\hat{S}(T)\right)}=\dfrac{Q_{n, M}\left(F(Z), \hat{S}(T)\right)}{U_n\left(\hat{S}(T)\right)}= \xi_{n,M}\left(F(Z), \hat{S}(T)\right) \notag
	\end{align}
	
	Next, we prove $\xi_{n,M}\left(F(Z), \hat{S}(T)\right) = \xi_{n,M}\left(F(Z), S(T)\right)$. Obviously, $\hat{S}(T)$ is with tie, then we breaking ties uniformly at random, which means $\hat{S}(T)$ can also be seen as strictly monotonically decreasing. Then according to the definition of $G_n(\cdot)$ and $F_n(\cdot)$, we have 
	\begin{align}
		& G_n\left(\hat{S}(t)\right) \equiv \frac{1}{n} \sum_{i=1}^n \mathds{1}\left(\hat{S}(T_i) \geq \hat{S}(t)\right), \quad  F_n\left(\hat{S}(t)\right) \equiv \frac{1}{n} \sum_{i=1}^n \mathds{1}\left(\hat{S}(T_i) \leq \hat{S}(t)\right) \notag \\
		& G_n\left(S(t)\right) \equiv \frac{1}{n} \sum_{i=1}^n \mathds{1}\left(S(T_i) \geq S(t)\right), \quad \quad F_n\left(S(t)\right) \equiv \frac{1}{n} \sum_{i=1}^n \mathds{1}\left(S(T_i) \leq S(t)\right) \notag 
	\end{align}
	that is $G_n\left(\hat{S}(t)\right) =  G_n\left(S(t)\right)$ and $F_n\left(\hat{S}(t)\right) = F_n\left(S(t)\right)$. Then
	\begin{align}
		Q_{n, M}\left(F(Z), \hat{S}(T)\right) \equiv 
		& \frac{1}{n M} \sum_{i=1}^n \sum_{m=1}^M  \min \left\{F_n\left(\hat{S}\left(T_i\right)\right), F_n\left(\hat{S}\left(T_{j_m\left(F\left(Z_i\right)\right)}\right)\right)\right\} -\frac{1}{n} \sum_{i=1}^n G_n\left(\hat{S}(T_i)\right)^2 \notag \\
		= & \frac{1}{n M} \sum_{i=1}^n \sum_{m=1}^M  \min \left\{F_n\left(S\left(T_i\right)\right), F_n\left(S\left(T_{j_m\left(F\left(Z_i\right)\right)}\right)\right)\right\} -\frac{1}{n} \sum_{i=1}^n G_n\left(S(T_i)\right)^2 \notag \\
		= & Q_{n, M}\left(F(Z), S(T)\right) \notag
	\end{align}
	Then according to the definition of $U(\cdot)$, we can also have
	\begin{align}
		U_n\left(\hat{S}(T)\right)
		& = \frac{1}{n} \sum_{i=1}^n G_n\left(\hat{S}(T_i)\right)\left(1-G_n\left(\hat{S}(T_i)\right)\right) \notag \\
		& = \frac{1}{n} \sum_{i=1}^n G_n\left(S(T_i)\right)\left(1-G_n\left(S(T_i)\right)\right)
		= U_n\left(S(T)\right) \notag
	\end{align}
	Then
	\begin{align}
		\xi_{n,M}\left(F(Z), \hat{S}(T)\right) = \dfrac{Q_{n, M}\left(F(Z), \hat{S}(T)\right)}{U_n\left(\hat{S}(T)\right)} = \dfrac{Q_{n, M}\left(F(Z), S(T)\right)}{U_n\left(S(T)\right)} = \xi_{n,M}\left(F(Z), S(T)\right) \notag
	\end{align}
	Thus we proved $\xi_{n,M}\left(\hat{F}(Z), \hat{S}(T)\right) = \xi_{n,M}\left(F(Z), \hat{S}(T)\right) = \xi_{n,M}\left(F(Z), S(T)\right)$.
	
	Next we prove, $\xi_{n,M}\left(\hat{S}(T), \hat{F}(Z)\right) = \xi_{n,M}\left(S(T), F(Z)\right)$. Since for any integer $i, j < n$
	\begin{align}
		\sum_{l=1}^n \mathds{1}\left(\hat{S}(T_{i})<\hat{S}(T_{l}) \leqslant \hat{S}(T_{j})\right) = \sum_{l=1}^n \mathds{1}\left(T_{i}<T_{l} \leqslant T_{j}\right) = \sum_{l=1}^n \mathds{1}\left(S(T_{i})<S(T_{l}) \leqslant S(T_{j})\right) \notag
	\end{align}
	we have $j_m\left(\hat{S}(T_i)\right)= j_m\left(S(T_i)\right)$. And since
	\begin{align}
		& G_n\left(\hat{F}(z)\right) \equiv \frac{1}{n} \sum_{i=1}^n \mathds{1}\left(\hat{F}(z_i) \geq \hat{F}(z)\right), \quad  F_n\left(\hat{F}(z)\right) \equiv \frac{1}{n} \sum_{i=1}^n \mathds{1}\left(\hat{F}(Z_i) \leq \hat{F}(z)\right) \notag \\
		& G_n\left(F(z)\right) \equiv \frac{1}{n} \sum_{i=1}^n \mathds{1}\left(F(Z_i) \geq F(z)\right), \quad \quad F_n\left(F(z)\right) \equiv \frac{1}{n} \sum_{i=1}^n \mathds{1}\left(F(Z_i) \leq F(z)\right) \notag 
	\end{align}
	we have $G_n\left(\hat{F}(z)\right) =  G_n\left(F(z)\right)$ and $F_n\left(\hat{F}(z)\right) = F_n\left(F(z)\right)$. Then
	\begin{align}
		Q_{n, M}\left(\hat{S}(T), \hat{F}(Z)\right) \equiv 
		& \frac{1}{n M} \sum_{i=1}^n \sum_{m=1}^M  \min \left\{F_n\left(\hat{F}\left(Z_i\right)\right), F_n\left(\hat{F}\left(Z_{j_m\left(\hat{S}\left(T_i\right)\right)}\right)\right)\right\} -\frac{1}{n} \sum_{i=1}^n G_n\left(\hat{F}(Z_i)\right)^2 \notag \\
		= & \frac{1}{n M} \sum_{i=1}^n \sum_{m=1}^M  \min \left\{F_n\left(F\left(Z_i\right)\right), F_n\left(F\left(Z_{j_m\left(S\left(T_i\right)\right)}\right)\right)\right\} -\frac{1}{n} \sum_{i=1}^n G_n\left(F(Z_i)\right)^2 \notag \\
		= & Q_{n, M}\left(S(T), F(Z)\right) \notag
	\end{align}
	And since
	\begin{align}
		U_n\left(\hat{F}(Z)\right)
		& = \frac{1}{n} \sum_{i=1}^n G_n\left(\hat{F}(Z_i)\right)\left(1-G_n\left(\hat{F}(Z_i)\right)\right) \notag \\
		& = \frac{1}{n} \sum_{i=1}^n G_n\left(F(Z_i)\right)\left(1-G_n\left(F(Z_i)\right)\right)
		= U_n\left(F(Z)\right) \notag
	\end{align}
	Then
	\begin{align}
		\xi_{n,M}\left(\hat{S}(T), \hat{F}(Z)\right) = \dfrac{Q_{n, M}\left(\hat{S}(T), \hat{F}(Z)\right)}{U_n\left(\hat{F}(Z)\right)} = \dfrac{Q_{n, M}\left(S(T), F(Z)\right)}{U_n\left(F(Z)\right)} = \xi_{n,M}\left(S(T), F(Z)\right) \notag
	\end{align}
	which end the proof.
	
\end{proof}


\begin{proof}[\textbf{Proof of the Theorem \ref{Sure screening property}:}]
	We first prove
	\begin{align}
		P\left(\left|\hat{\omega}_k-\omega_k\right| \geqslant \epsilon \right) \leqslant
		8 \exp \left(-\Gamma_3 \frac{\epsilon^2 n^2-4\epsilon   }{576 n}\right) \notag
	\end{align}
	
	In the following part, for writing convincing, $\xi_{n, M}\left(\hat{F}_{X_k}\left(X_k\right), \hat{S}(T)\right)$ shorthand as $\xi_{n, M}(\hat{F}_k, \hat{S})$, and $\hat{F}_{X_k}\left(X_k\right)$ and $\hat{S}(T)$ shorthand as $\hat{F}_k$ and $\hat{S}$ respectively. Since
	\begin{align}
		\hat{w}_k \equiv & \max\Bigl\{\xi_{n,M}(\hat{F}_k, \hat{S}), \ \xi_{n,M}(\hat{S}, \hat{F}_k) \Bigr\} \notag \\
		\equiv & \frac{1}{2}\Bigl(\xi_{n,M}(\hat{F}_k, \hat{S}) + \xi_{n,M}(\hat{S}, \hat{F}_k) + |\xi_{n,M}(\hat{F}_k, \hat{S}) - \xi_{n,M}(\hat{S}, \hat{F}_k)|\Bigr) , \quad k=1, \ldots, p \notag \\
		w_k \equiv & \max\Bigl\{\xi({F}_k, {S}), \xi({S}, {F}_k) \Bigr\} \notag \\ 
		\equiv & \frac{1}{2}\Bigl(\xi({F}_k,{S}) + \xi({S}, {F}_k) + |\xi({F}_k, {S}) - \xi({S}, {F}_k)|\Bigr) , \quad k=1, \ldots, p \notag
	\end{align}
	Then
	\begin{align}
		& P\left(\left|\hat{\omega}_k-\omega_k\right| \geqslant \epsilon \right) \notag \\ 
		= & P\left(\left|\max\Bigl\{\xi_{n,M}(\hat{F}_k, \hat{S}), \ \xi_{n,M}(\hat{S}, \hat{F}_k) \Bigr\} - \max\Bigl\{\xi({F}_k, {S}), \xi({S}, {F}_k) \Bigr\} \right| \geqslant \epsilon \right) \notag \\
		= & P\Bigl( \Bigl| \xi_{n,M}(\hat{F}_k, \hat{S}) + \xi_{n,M}(\hat{S}, \hat{F}_k) + |\xi_{n,M}(\hat{F}_k, \hat{S}) - \xi_{n,M}(\hat{S}, \hat{F}_k)|  \notag \\
		& \quad - \xi({F}_k,{S}) - \xi({S}, {F}_k) - |\xi({F}_k, {S}) - \xi({S}, {F}_k)| \Bigr| \geqslant 2\epsilon \Bigr) \notag \\
		\leqslant & P\left(\left|\xi_{n,M}(\hat{F}_k, \hat{S})-\xi({F}_k,{S})\right| \geqslant \frac{2\epsilon}{3} \right) + P\left(\left|\xi_{n,M}(\hat{S}, \hat{F}_k)-\xi({S}, {F}_k)\right| \geqslant \frac{2\epsilon}{3} \right) \notag \\
		& \quad + P\left(\left||\xi_{n,M}(\hat{F}_k, \hat{S}) - \xi_{n,M}(\hat{S}, \hat{F}_k)|-|\xi({F}_k, {S}) - \xi({S}, {F}_k)|\right| \geqslant \frac{2\epsilon}{3} \right) \notag \\
		\leqslant & P\left(\left|\xi_{n,M}(\hat{F}_k, \hat{S})-\xi({F}_k,{S})\right| \geqslant \frac{2\epsilon}{3} \right) + P\left(\left|\xi_{n,M}(\hat{S}, \hat{F}_k)-\xi({S}, {F}_k)\right| \geqslant \frac{2\epsilon}{3} \right) \notag \\
		& \quad + P\left(\left|\xi_{n,M}(\hat{F}_k, \hat{S}) - \xi_{n,M}(\hat{S}, \hat{F}_k)-\xi({F}_k, {S}) + \xi({S}, {F}_k)\right| \geqslant \frac{2\epsilon}{3} \right) \notag \\
		\leqslant & P\left(\left|\xi_{n,M}(\hat{F}_k, \hat{S})-\xi({F}_k,{S})\right| \geqslant \frac{2\epsilon}{3} \right) + P\left(\left|\xi_{n,M}(\hat{S}, \hat{F}_k-\xi({S}, {F}_k)\right| \geqslant \frac{2\epsilon}{3} \right) \notag \\
		& \quad + P\left(\left|\xi_{n,M}(\hat{F}_k, \hat{S})-\xi({F}_k, {S}) \right| \geqslant \frac{\epsilon}{3} \right)
		+ P\left(\left|\xi_{n,M}(\hat{S}, \hat{F}_k)-\xi({S}, {F}_k) \right| \geqslant \frac{\epsilon}{3} \right) \notag \\
		\leqslant &  2P\left(\left|\xi_{n,M}(\hat{F}_k, \hat{S})-\xi({F}_k, {S}) \right| \geqslant \frac{\epsilon}{3} \right)+ 2P\left(\left|\xi_{n,M}(\hat{S}, \hat{F}_k)-\xi({S}, {F}_k) \right| \geqslant \frac{\epsilon}{3} \right) \notag \\
		= &  2P\left(\left|\xi_{n,M}(F_k, S)-\xi(F_k, S) \right| \geqslant \frac{\epsilon}{3} \right)+ 2P\left(\left|\xi_{n,M}(S, F_k)-\xi(S, F_k) \right| \geqslant \frac{\epsilon}{3} \right) \notag \\
		\leqslant & 2\cdot 2 \exp \left(-\Gamma_1 \frac{\epsilon^2 n^2-4\epsilon   }{576 n}\right)+ 2\cdot 2 \exp \left(-\Gamma_2 \frac{\epsilon^2 n^2-4\epsilon   }{576 n}\right) \notag \\
		= & 8 \exp \left(-\Gamma_3 \frac{\epsilon^2 n^2-4\epsilon   }{576 n}\right)\notag
	\end{align}
	where $\Gamma_3=\min \{\Gamma_1, \Gamma_2\}$. The first and the third inequality holds due to the triangle inequality $|x \pm y| \leqslant |x|+|y|$, the second inequality holds due to the reverse triangle inequality $||x|-|y|| \leqslant |x-y|$. The last second equality holds due to lemma \ref{lemma.2} and the last inequality holds due to lemma \ref{lemma.1}
	
	Let $\epsilon=c n^{-\kappa}$, where $\kappa$ satisfies $0 \leqslant \kappa<1 / 2$. We thus have
	\begin{align}
		& P\left(\max _{1 \leqslant k \leqslant p}\left|\hat{\omega}_k-\omega_k\right| \geqslant c n^{-\kappa}\right) 
		\leqslant p \max _{1 \leqslant k \leqslant p} P\left(\left|\hat{\omega}_k-\omega_k\right| \geqslant c n^{-\kappa} \right) \notag \\
		\leqslant & 8p \exp \left(  -\frac{\Gamma_3 c^2 n^{1-2\kappa}}{576}  + \frac{ \Gamma_3 c n^{-\kappa-1} }{144} \right) 
		\leqslant 8p \exp \left(  -\frac{\Gamma_3 c^2 n^{1-2\kappa}}{576}  + \frac{ \Gamma_3 c }{144} \right) \notag \\
		= & 8 \exp \left( \frac{ \Gamma_3 c }{144} \right) p \exp \left(-\frac{\Gamma_3 c^2 n^{1-2\kappa}}{576} \right) 
		=  O\left(p \exp \left\{-\Gamma n^{1-2 \kappa}\right\}\right) \notag
	\end{align}
	where $\Gamma = \Gamma_3 c^2/576$.
	
	Thus the first part of Theorem \ref{Sure screening property} is proved. To prove the second part of Theorem \ref{Sure screening property}, we consider the event
	$$
	\mathcal{B}=\left\{\max _{k \in \mathcal{A}}\left|\hat{\omega}_k-\omega_k\right| \leqslant c n^{-\kappa}\right\} .
	$$
	Since for all $k \in \mathcal{A}$, the condition (C1) ensures that $\omega_k \geqslant 2 c n^{-\kappa}$, then for event $\mathcal{B}$, we have for all $k \in \mathcal{A}$, $\hat{\omega}_k \geqslant c n^{-\kappa}$. Hence we have $\mathcal{A} \subseteq \hat{\mathcal{A}}$. Thus we get
	\begin{align}
		P(\mathcal{A} \subseteq \hat{\mathcal{A}}) \geqslant P(\mathcal{B})=1-P\left(\max _{k \in \mathcal{A}}\left|\hat{\omega}_k-\omega_k\right|>c n^{-\kappa}\right)= 1-O\left(s \exp \left\{-\Gamma n^{1-2 \kappa}\right\}\right) \notag
	\end{align}
\end{proof}

\begin{proof}[\textbf{Proof of the Theorem \ref{Controlling false discoveries}:}]
	We consider the set
	$\mathcal{C}=\left\{\max _{1 \leqslant k \leqslant p}\left|\hat{\omega}_k-\omega_k\right| \leqslant 2^{-1} c n^{-\kappa}\right\}$
	on the set $\mathcal{C}$, the number of $\left\{k: \hat{\omega}_k \geqslant c n^{-\kappa}\right\}$ cannot exceed the number of $\left\{k: \omega_k \geqslant 2^{-1} c n^{-\kappa}\right\}$.
	
	Further note that the number of $\left\{k: \omega_k \geqslant 2^{-1} c n^{-\kappa}\right\}$ cannot exceed $2 c^{-1} n^\kappa \sum_k \omega_k$. Otherwise, we have
	$$
	\sum_k \omega_k \geqslant\left(1+2 c^{-1} n^\kappa \sum_k \omega_k\right) \times \frac{c}{2} n^{-\kappa}>\sum_k \omega_k .
	$$    
	Thus, for $\hat{\mathcal{A}}=\left\{k: \hat{\omega}_k \geqslant c n^{-\kappa} \text {, for } 1 \leqslant k \leqslant p\right\}$, we have
	$$
	P\left(|\hat{\mathcal{A}}| \leqslant 2 c^{-1} n^\kappa \sum_k \omega_k\right) \geqslant P(\mathcal{C}) \geqslant 1-O\left(p \exp \left(-\Gamma n^{1-2 \kappa}\right)\right) .
	$$
\end{proof}


\begin{proof}[\textbf{Proof of the Theorem \ref{Theorem 3}:}]
	Recalling the condition (C2), $\min _{k \in \mathcal{A}} \omega_k-\max _{k \in \mathcal{I}} \omega_k=m>0$. Thus we have
	\begin{align}
		& P\left(\min _{k \in \mathcal{A}} \hat{\omega}_k \leqslant \max _{k \in \mathcal{I}} \hat{\omega}_k\right) \notag \\
		= & P\left(\max _{k \in \mathcal{I}} \hat{\omega}_k - \min _{k \in \mathcal{A}} \hat{\omega}_k + m \geqslant m \right) \notag \\
		= & P\left(\left[\max _{k \in \mathcal{I}} \hat{\omega}_k-\max _{k \in \mathcal{I}} \omega_k\right]-\left[\min _{k \in \mathcal{A}} \hat{\omega}_k-\min _{k \in \mathcal{A}} \omega_k\right] \geqslant m \right) \notag \\
		\leqslant & P \left(\max _{k \in \mathcal{I}}\left|\hat{\omega}_k-\omega_k\right|+\max _{k \in \mathcal{A}}\left|\hat{\omega}_k-\omega_k\right| \geqslant m \right) \notag \\ 
		\leqslant & P\left(\max _{1 \leqslant k \leqslant p}\left|\hat{\omega}_k-\omega_k\right| \geqslant \frac{m}{2}\right) \leqslant O\left(p \exp \left\{-\Gamma n m^2\right\}\right) \notag
	\end{align}
	This completes the proof.
\end{proof}

\newpage




 \bibliographystyle{elsarticle-harv} 
 \bibliography{references_xsis}



\end{document}